\documentclass[a4paper,11pt]{article}
\usepackage{listings}
  \usepackage{mathrsfs}
       \usepackage[mathscr]{euscript}
\newcommand{\No}{{\sc No}}
\newcommand{\Yes}{{\sc Yes}}
\newcommand{\FPT}{{\sf FPT}}
\newcommand{\OO}{{\mathcal{O}}}

\newcommand{\alg}[1]{\mbox{\sf #1}}  

\newcommand{\boundS}{$s\geq 2r2^c+r$}

\usepackage{}
\usepackage{amsfonts,amsthm,amsmath}
\usepackage{amstext}
\usepackage{graphicx,tikz}
\RequirePackage{fancyhdr}
\usepackage{xcolor}
\usepackage{boxedminipage}
\newcommand{\olddaniel}[1]{}

\usepackage{listings}

\usepackage{vmargin}
\setmarginsrb{1in}{1in}{1in}{1in}{0pt}{0pt}{0pt}{6mm}

\newcommand{\bran}[1]{branchable\xspace}

\usepackage[linesnumbered, boxed, algosection,noline]{algorithm2e}

\newcommand{\myparagraph}[1]{\smallskip\noindent{\textbf{\sffamily #1} \ }}

\newtheorem{theorem}{Theorem}
\newtheorem{lemma}{Lemma}[section]
\newtheorem{claim}{Claim}[section]

\newtheorem{definition}{Definition}[section]
\newtheorem{observation}{Observation}[section]

\date{}

\usepackage{graphics}

\usepackage{pgf}
\usepackage{paralist} 
\usepackage{graphicx}
\usepackage{boxedminipage}
\usepackage{boxedminipage}
\usepackage{xcolor}
\usepackage[linesnumbered,boxed,algosection]{algorithm2e}
\usepackage{framed}

\usepackage{amsmath, amssymb, latexsym}
\usepackage{enumerate}

\usepackage{float}
\usepackage{xspace}

\usepackage{todonotes}

\usepackage{booktabs}

 \usepackage[ pdftex, plainpages = false, pdfpagelabels, 
                 bookmarks=false,
                 bookmarksopen = true,
                 bookmarksnumbered = true,
                 breaklinks = true,
                 linktocpage,
                 pagebackref,
                 colorlinks = true,  
                 linkcolor = blue,
                 urlcolor  = blue,
                 citecolor = red,
                 anchorcolor = green,
                 hyperindex = true,
                 hyperfigures
                 ]{hyperref}

\usepackage{thm-restate}

\usepackage{thmtools}

\begin{document}


\title{Reducing CMSO Model Checking to Highly Connected Graphs\thanks{Supported by {\em Pareto-Optimal Parameterized Algorithms}, ERC Starting Grant 715744 and {\em Parameterized Approximation}, ERC Starting Grant 306992.
%
%
%
          M. S. Ramanujan also acknowledges support from {\em BeHard},  Bergen Research Foundation and {\em X-Tract}, Austrian Science Fund (FWF, project P26696).} }

 \author{
 Daniel Lokshtanov\thanks{University of Bergen, Bergen, Norway. \texttt{daniello@ii.uib.no}}
 \and M. S. Ramanujan\thanks{University of Warwick, UK.     \texttt{R.Maadapuzhi-Sridharan@warwick.ac.uk}} 
 \and  Saket Saurabh\thanks{The Institute of Mathematical Sciences, HBNI, Chennai, India. \texttt{saket@imsc.res.in}}\hspace{2pt} \addtocounter{footnote}{-3}\footnotemark \addtocounter{footnote}{2}
 \and Meirav Zehavi\thanks{Ben-Gurion University, Israel. \texttt{Zehavimeirav@gmail.com}}
 }
 
%

 \maketitle

\thispagestyle{empty}
\begin{abstract} 
Given a Counting Monadic Second Order (CMSO) sentence $\psi$, the {\sc CMSO}$[\psi]$ problem is defined as follows. The input to {\sc CMSO}$[\psi]$ is a graph $G$, and the objective is to determine whether $G\models \psi$. Our main theorem states that for every CMSO sentence  $\psi$, if {\sc CMSO}$[\psi]$ is solvable in polynomial time on ``globally highly connected graphs'', then {\sc CMSO}$[\psi]$ is solvable in polynomial time (on general graphs). We demonstrate the utility of our theorem in the design of parameterized algorithms. Specifically we show that technical problem-specific ingredients of a powerful method for designing parameterized algorithms,  recursive understanding, can be replaced by a black-box invocation of our main theorem. We also show that our theorem can be easily  deployed to show fixed parameterized tractability of a wide range of problems, where the input is a graph $G$ and the task is to find a connected induced subgraph of $G$ such that ``few'' vertices in this subgraph have neighbors outside the subgraph, and additionally the subgraph has a CMSO-definable property.

\end{abstract}

\newpage
\pagestyle{plain}
\setcounter{page}{1}

\section{Introduction}\label{sec:intro}

Algorithmic meta-theorems are general algorithmic results applicable to a whole range of problems. Many prominent algorithmic meta-theorems are about model checking;
such theorems state that for certain kinds of logic $L$, and all classes ${\cal C}$ that have a certain property, there is an algorithm that takes as input a formula $\phi \in L$ and a structure $S \in {\cal C}$ and efficiently determines whether $S \models \phi$. 
Results in this direction include the seminal theorem of Courcelle~\cite{Courcelle90,Courcelle92a,Courcelle97} for model checking Monadic Second Order Logic (MSO) on graphs of bounded treewidth (see also~\cite{AbrahamsonF93,ArnborgLS91,BoriePT92,
Courcelle:2012book,DowneyF98}), as well as a large body of work on model checking first-order (FO) logic~\cite{BovaGS16,DawarGK07,DvorakKT10,FlumG01,FrickG01,GHKOST13,GajarskyHOO14,gks14,Seese96}.

Another kind of algorithmic meta-theorems {\em reduce} the task of designing one type of algorithm for a problem, to one of designing a different kind of algorithm for the same problem.
The hope is, of course, that the second type of algorithms are significantly easier to design than the first. 
A prototype example of such results is {\em Bidimensionality}~\cite{DemaineFHT05}, which reduces the design of sub-exponential time parameterized algorithms for a problem on planar (or $H$-minor free) graphs, to the design of single exponential time algorithms for the same problem when parameterized by the treewidth of the input graph.

In this paper we prove a result of the second type for model checking  Counting Monadic Second Order Logic (CMSO), an extension of MSO with atomic sentences for determining the cardinality of vertex and edge sets modulo any (fixed) integer. For every  CMSO sentence $\psi$ define the {\sc CMSO}$[\psi]$ problem as follows. The input is a graph $G$ on $n$ vertices, and the task is to determine whether $G \models \psi$.

Our main result states that for every CMSO sentence $\psi$, if there is a $\OO(n^d)$ time algorithm ($d>4$) for {\sc CMSO}$[\psi]$ for the special case when the input graph is required to be ``highly connected everywhere'', then there is a $\OO(n^{d})$ time algorithm for {\sc CMSO}$[\psi]$ without any restrictions. In other words, our main theorem reduces CMSO model checking to model checking the same formula on graphs which are ``highly connected everywhere''.

In order to complete the description of our main result we need to define what we mean by ``highly connected everywhere''. For two integers $s$ and $c$, we say that a graph $G$ is $(s, c)$-{\em unbreakable} if there does not exist a partition of the vertex set into three sets $X$, $C$, and $Y$ such that 
\begin{itemize}
\item $C$ is a separator: there are no edges from $X$ to $Y$,
\item $C$ is small: $|C| \leq c$, and
\item $X$ and $Y$ are large: $|X|,|Y| \geq s$.
\end{itemize}
For example, the set of $(1, c)$-unbreakable graphs contains precisely the $(c+1)$-connected graphs, i.e. the connected graphs for which removing any set of at most $c$ vertices leaves the graph connected. We can now state our main result:

\begin{restatable}{theorem}{maintheoremgraphs}
\label{thm:main_graphs}
Let $\psi$ be a CMSO sentence. For all $c\in\mathbb{N}$, there exists $s\in\mathbb{N}$ such that if there exists an algorithm that solves {\sc CMSO}$[\psi]$ on $(s,c)$-unbreakable graphs in time $\OO(n^d)$ for some $d>4$, then there exists an algorithm that solves {\sc CMSO}$[\psi]$ on general graphs in time $\OO(n^{d})$.
\end{restatable}

%

For Theorem~\ref{thm:main_graphs} to be useful,  there must exist problems that can be formulated in CMSO, for which it is easier to design algorithms for the special case when the input graphs are unbreakable, than it is to design algorithms that work on general graphs. 
Such problems can be found in abundance in {\em parameterized complexity}. Indeed, the 
{\em recursive understanding} technique, which has been used to solve several  
problems~\cite{DBLP:journals/siamcomp/ChitnisCHPP16,DBLP:conf/stoc/GroheKMW11,DBLP:conf/focs/KawarabayashiT11,KimPST17,0001RS16,0001R16} in parameterized complexity, is based precisely on the observation that for many graph problems it is much easier to design algorithms if the input graph can be assumed to be unbreakable.

Designing algorithms using the recursive understanding technique typically involves a technical and involved argument akin to doing dynamic programming on graphs of bounded treewidth (see Chitnis et al.~\cite{DBLP:journals/siamcomp/ChitnisCHPP16} for an exposition).
These arguments reduce the original problem on general graphs to a generalized version of the problem on $(s,c)$-unbreakable graphs, for appropriate values of $s$ and $c$. Then an algorithm is designed for this generalized problem on $(s,c)$-unbreakable graphs, yielding an algorithm for the original problem.

For all applications of the recursive understanding technique known to the authors~\cite{DBLP:journals/siamcomp/ChitnisCHPP16,DBLP:conf/stoc/GroheKMW11,DBLP:conf/focs/KawarabayashiT11,KimPST17,0001RS16,0001R16}, the problem in question (in which recursive understanding has been applied) can be formulated as a CMSO model checking problem, and therefore, the rather cumbersome application of recursive understanding can be completely replaced by a black box invocation of Theorem~\ref{thm:main_graphs}.
Using Theorem~\ref{thm:main_graphs} in place of recursive understanding has the additional advantage that it reduces problems on general graphs to {\em the same} problem on unbreakable graphs, facilitating also the last step of designing an algorithm on unbreakable graphs. 

As an example of the power of Theorem~\ref{thm:main_graphs} we use it to give a fixed parameter tractable ({\FPT}) algorithm for the {\sc Vertex Multiway Cut Uncut} problem. Here, we are given a  graph $G$ together with a set of terminals 
$T \subseteq V(G)$, an equivalence relation $\mathcal{R}$ on the set $T$, and an integer $k$, and the objective is to test 
whether there exists a set $S\subseteq V (G) \setminus T$ of at most $k$ vertices such that for any $u, v \in T$, 
the vertices $u$ and $v$ belong to the same connected component of $G\setminus S$ if and only if
 $(u, v)\in \mathcal{R}$. Since finding the desired set $S$ satisfying the above property can be formulated in  CMSO, we are able to completely sidestep the necessity to define a technically involved annotated version of our problem, and furthermore, we need only focus on the base case where the graph is unbreakable. To solve the base case, a simple procedure that is based on the enumeration of connected sets with small neighborhood  is sufficient. For classification purposes, our approach is significantly simpler than the problem-specific algorithm in~\cite{DBLP:journals/siamcomp/ChitnisCHPP16}.  Finally, we show how Theorem~\ref{thm:main_graphs} can be effortlessly deployed to show fixed parameterized tractability of a  wide range of problems, where the input is a graph $G$ and the task is to find a connected induced subgraph of $G$ of bounded treewidth such that ``few'' vertices outside this subgraph have neighbors inside the subgraph, and additionally the subgraph has a CMSO-definable property.

 \myparagraph{Our techniques.} The proof of Theorem \ref{thm:main_graphs} is based heavily on 
 the idea of
graph replacement, which dates back to the work of Fellows and Langston~\cite{FellowsL89}. We combine this idea with  Courcelle's theorem~\cite{Courcelle90,Courcelle92a,Courcelle97}, which states that every CMSO-definable property $\sigma$   has finite state on a bounded-size separation/boundary. In other words, for any CMSO-definable property $\sigma$ and fixed $t\in {\mathbb N}$, there is an equivalence relation  defined on the set of all $t$-boundaried graphs (graphs with a set of at most $t$ distinguished vertices) with a finite number, say $\zeta$ (where $\zeta$ depends only on $\sigma$ and $t$) of equivalence classes such that if we replace any $t$-boundaried subgraph $H$ of the given graph $G$ with another $t$-boundaried graph, say $H'$,  from the same equivalence class to obtain a graph $G'$, then $G$ has the property $\sigma$ if and only if $G'$ has the property $\sigma$. In our case, $t=2c$. Let   $R_1,\dots, R_\zeta$ denote a set containing one `` minimal'' $2c$-boundaried graph    from each  equivalence class (for the fixed CMSO-definable property $\sigma$). Let $r$ denote the size of the largest among these minimal representatives.

The main technical content of our paper is in the description of an algorithm for a generalization of our question. To be precise, we will describe how one can, given a $2c$-boundaried graph $G$, locate the precise equivalence class in which $G$ is contained and how one could compute the corresponding smallest representative from the set $\{R_1,\dots, R_\zeta\}$. We refer to this  task as ``understanding'' $G$.
%

In order to achieve our objective, we first give an algorithm $\cal A$ that allows one to understand $2c$-boundaried $(s-r,c)$-unbreakable graphs (for a choice of $s$ which is sufficiently large compared to $r$ and $c$). This algorithm is built upon the following observation. The equivalence class of any $2c$-boundaried graph $G$ is determined exactly by the subset of $\{G\oplus R_1,G\oplus R_2,\dots, G\oplus R_\zeta\}$ on which $\sigma$ evaluates to true. Here, the graph $G\oplus R_i$ is the graph obtained by taking the disjoint union of the graphs $G$ and $R_i$ and then identifying the vertices of the boundaries of these graphs with the same label. Since $s$ is chosen to be sufficiently large compared to $c$ and $r$, it follows that for every $i\in \{1,\dots, \zeta\}$, the graph $G\oplus R_i$ is $(s,c)$-unbreakable and we can use the assumed algorithm for CMSO[$\psi$] on $(s,c)$-unbreakable graphs to design an algorithm that \emph{understands} $2c$-boundaried $(s-r,c)$-unbreakable graphs. This constitutes the `base case' of our main algorithm. 

In order to understand a general ($(s-r,c)$-breakable) $2c$-boundaried graph, we use 
 known algorithms from \cite{DBLP:journals/siamcomp/ChitnisCHPP16} to compute a partition of the vertex set of $G$ into $X,C,$ and $Y$ such that $C$ is a separator, $|C|\leq c$ and $|X|,|Y|\geq 
\frac{s-r}{2^c}$. Let $G_1=G[X\cup C]$ and let $G=G[Y\cup C]$. 
 Without loss of generality, we may assume that at most half the vertices in the boundary of $G$ lie in $X\cup C$. Consequently, the graph $G_1$ is a $2c$-boundaried graphs where the boundary vertices are the vertices in $C$ along with the boundary vertices of $G$ contained in $X\cup C$. We then recursively understand the strictly smaller $2c$-boundaried graph $G_1$ to find its representative $\hat R\in \{R_1,\dots, R_\zeta\}$. Since the evaluation of $\sigma$ on $G$ is the same as the evaluation of $\sigma$ on $G_2\oplus \hat R$ (where the gluing happens along $C$), we only need to understand the $2c$-boundaried graph $G_2\oplus \hat R$ (where the boundary is carefully defined from that of $G$ and $\hat R$) and we do this by recursively executing the ``understand'' algorithm on this graph.

%
%
%
%

At this point we also need to remark on two drawbacks of Theorem~\ref{thm:main_graphs}. The first is that Theorem~\ref{thm:main_graphs} is {\em non-constructive}. Given an algorithm for {\sc CMSO}$[\psi]$ on $(s,c)$-unbreakable graphs, Theorem~\ref{thm:main_graphs} allows us to infer the existence of an algorithm for {\sc CMSO}$[\psi]$ on general graphs, but it does not provide us with the actual algorithm. This is due to  the subroutine $\cal S$ requiring a representative $2c$-boundaried subgraph for each equivalence class, to be part of its `source code'.  
Thus, the parameterized algorithms obtained  using Theorem~\ref{thm:main_graphs} are {\em non-uniform} (see Section~\ref{sec:applications}), as opposed to the algorithms obtained by recursive understanding.

 The second drawback is that Theorem~\ref{thm:main_graphs} incurs a gargantuan constant factor overhead in the running time, where this factor depends on the formula $\psi$ and the cut size $c$. We leave removing these two drawbacks as intriguing open problems.

\section{Preliminaries}\label{sec:prelims}

In this section, we introduce basic terminology related to graphs, structures, CMSO, boundaried structures and parameterized complexity.  
In order to present a rigorous proof of our lemmas in a way that is consistent with existing notation used in related work, we 
 follow the notation from the paper~\cite{DBLP:journals/jacm/BodlaenderFLPST16}.
We use $[t]$ as a shorthand for $\{1,2,\ldots,t\}$. Given a function $f: A\rightarrow B$ and a subset $A'\subseteq A$, we denote $f(A')=\bigcup_{a\in A'}f(a)$.

\subsection{Graphs}\label{sec:prelimsGraphs}

Throughout this paper, we use the term ``graph'' to refer to a multigraph rather than only a simple graph. Given a graph $G$, we let $V(G)$ and $E(G)$ denote the vertex and edge sets of $G$, respectively. When $G$ is clear from the context, we denote $n=|V(G)|$ and $m=|E(G)|$. Given two subsets of $V(G)$, $A$ and $B$, we let $E(A,B)$ denote the set of edges of $G$ with one endpoint in $A$ and the other endpoint in $B$. Given $U\subseteq V(G)$, we let $G[U]$ denote the subgraph of $G$ induced by $U$, and we let $N(U)$ and $N[U]$ denote the open and closed neighborhoods of $U$, respectively. Moreover, we denote $G\setminus U=G[V(G)\setminus U]$. Given $v\in V(G)$, we denote $N(v)=N(\{v\})$ and $N[v]=N[\{v\}]$. Given $E\subseteq E(G)$, we denote $G\setminus E = (V(G),E(G)\setminus E)$. Moreover, we let $V[E]$ denote the set of every vertex in $V(G)$ that is incident to at least one edge in $E$, and we define $G[E]=(V[E],E)$. A graph $G$ is a {\em cluster graph} if there exists a partition $(V_1,V_2,\ldots,V_r)$ of $V(G)$ for some $r\in\mathbb{N}_0$ of $V(G)$ such that for all $i\in[r]$, $G[V_i]$ is a clique, and for all $j\in[r]\setminus\{i\}$, $E(V_i,V_j)=\emptyset$.

\bigskip
{\noindent\bf Treewidth.} Treewidth is a structural parameter that indicates how much a given graph resembles a tree. For example, a tree has treewidth 1 and an $n$-vertex clique has treewidth $n-1$. Formally, the treewidth of a graph is defined as follows.

\begin{definition}\label{def:treewidth}
A \emph{tree decomposition} of a graph $G$ is a pair $(T,\beta)$ of a tree $T$ and $\beta: V(T) \rightarrow 2^{V(G)}$,
such that
\begin{enumerate}
    \item $\bigcup_{t \in V(T)} \beta(t) = V(G)$, and
    \item\label{item:twedge} for any edge $e \in E(G)$, there exists a node $t \in V(T)$ such that both endpoints of $e$ belong to $\beta(t)$, and
    \item\label{item:twconnected} for any vertex $v \in V(G)$, the subgraph of $T$ induced by the set $T_v = \{t\in V(T): v\in\beta(t)\}$ is a tree.
\end{enumerate}
The {\em width} of $(T,\beta)$ is $\max_{v\in V(T)}\{|\beta(v)|\}-1$. The {\em treewidth} of $G$ is the minimum width of a tree decomposition of $G$.
\end{definition}

{\noindent\bf Unbreakability.} To formally introduce the notion of unbreakability, we rely on the definition of a separation:

\begin{definition}{\rm [\bf Separation]} A pair $(X, Y)$ where $X \cup Y = V (G)$ is a {\em separation} if $E(X \setminus
Y, Y \setminus X) = \emptyset$. The order of $(X, Y)$ is $\vert X\cap Y\vert$.
\end{definition}

Roughly speaking, a graph is breakable if it is possible to ``break'' it into two large parts by removing only a small number of vertices. Formally,

\begin{definition}{\rm [\bf $(s,c)$-Unbreakable graph]} Let $G$ be a graph. If there exists a separation $(X,Y)$ of order at most $c$ such that $\vert X\setminus Y\vert > s$ and $\vert Y\setminus X\vert> s$, called an {\em $(s,c)$-witnessing separation}, then $G$ is {\em $(s,c)$-breakable}. Otherwise, $G$ is {\em $(s,c)$-unbreakable}.
\end{definition}

The following lemma implies that it is possible to determine (approximately) whether a graph is unbreakable or not, and lemmata similar to it can be found in \cite{DBLP:journals/siamcomp/ChitnisCHPP16}. We give a proof for the sake of completeness and in order to avoid interrupting the flow of the reader and to ensure that we keep the presentation focussed on the main result, the proof has been moved to the Appendix (Section~\ref{app:break}).

\begin{lemma}\label{lem:break}
There exists an algorithm, \alg{Break-ALG}, that given $s,c\in\mathbb{N}$ and a graph $G$, in time $2^{\OO(c\log(s+c))}\cdot n^3\log n$ either returns an $\displaystyle{(\frac{s}{2^c},c)}$-witnessing separation or correctly concludes that $G$ is $(s,c)$-unbreakable.
\end{lemma}

{\noindent\bf Boundaried Graphs.} Roughly speaking, a boundaried graph is a graph where some vertices are labeled. Formally,

\begin{definition}{\rm [\bf Boundaried graph]}\label{def:boundariedGraph}
A {\em boundaried graph} is a graph $G$ with a set $\delta(G)\subseteq V(G)$ of distinguished vertices called {\em boundary vertices}, and an injective labeling $\lambda_G: \delta(G)\rightarrow \mathbb{N}$. The set $\delta(G)$ is the {\em boundary} of $G$, and the {\em label set} of $G$ is $\Lambda(G)=\{\lambda_{G}(v)\mid v\in \delta(G)\}$.
\end{definition}

We remark that we also extend the definition of $(s,c)$-(un)breakability from graphs, to boundaried graphs in the natural way. That is, we ignore the boundary vertices when considering the existence of an $(s,c)$-witnessing separation.  For ease of presentation, we sometimes abuse notation and treat equally-labeled vertices of different boundaried graphs, as well as the vertex that is the result of the identification of two such vertices, as the same vertex. Given a finite set $I\subseteq \mathbb{N}$, ${\cal F}_{I}$  denotes the class of all boundaried graphs whose label set is $I$, and ${\cal F}_{\subseteq I}=\bigcup_{I'\subseteq I}{\cal F}_{I'}$. A boundaried graph in ${\cal F}_{\subseteq [t]}$ is called  a {\em $t$-boundaried} graph. Finally, ${\cal F}$ denotes the class of all boundaried graphs. The main operation employed to unite two boundaried graphs is the one that glues their boundary vertices together. Formally,

\begin{definition}{\rm [\bf Gluing by $\oplus$]} Let $G_1$ and $G_2$ be two  boundaried graphs. Then, $G_1 \oplus G_2$ is the (not-boundaried) graph  obtained from the disjoint union of $G_1$ and $G_2$ by identifying equally-labeled vertices in $\delta(G_1)$ and $\delta(G_2)$.\footnote{Each edge in $G_1$ (or $G_2$) whose endpoints are boundaried vertices in $G_1$ (or $G_2$) is preserved as a unique edge in $G_1 \oplus G_2$.}
\end{definition}


\subsection{Structures}\label{sec:prelimsStruct}

We first define the notion of a {\em structure} in the context of our paper. 

\begin{definition}{\rm [{\bf Structure}]}
A {\em structure} $\alpha$ is a tuple whose first element is a graph, denoted by $G_\alpha$, and each of the remaining elements is a subset of $V(G_\alpha)$, a subset of $E(G_\alpha)$, a vertex in $V(G_\alpha)$ or an edge in $E(G_\alpha)$. The number of elements in the tuple is the {\em arity} of the structure.
\end{definition}

Given a structure $\alpha$ of arity $p$ and an integer $i\in [p]$, we let $\alpha[i]$ denote the $i$'th element of $\alpha$. Note that $\alpha[1]=G_\alpha$.
By {\em appending} a subset $S$ of $V(G_{\alpha})$ (or $E(G_{\alpha})$) to a structure $\alpha$ of arity $p$, we produce a new structure, denoted by $\alpha' = \alpha\diamond S$, of arity $p+1$ with the first $p$ elements of $\alpha'$ being the elements of $\alpha$ and $\alpha'[p+1] = S.$
For example, consider the structure $\alpha=(G_{\alpha},S,e)$ of arity $3$, where $S\subseteq V(G_{\alpha})$ and $e\in E(G_{\alpha}).$
Let $S'$ be some subset of $V(G_{\alpha})$. Then, appending $S'$ to $\alpha$ results in the structure $\alpha'=\alpha\diamond S'=(G_{\alpha},S,e,S').$

Next, we define the notions of a {\em type} of a structure and a {\em property} of structures. 

\begin{definition}{\rm [{\bf Type}]}
Let $\alpha$ be a structure of arity $p$. The {\em type} of $\alpha$ is a tuple of arity $p$, denoted by ${\bf type}(\alpha)$, where the first element, ${\bf type}(\alpha)[1]$, is {\sf graph}, 
and for every $i\in\{2,3,\ldots,p\}$, ${\bf type}(\alpha)[i]$ is {\sf vertex} if $\alpha[i]\in V(G_\alpha)$, {\sf edge} if $\alpha[i]\in E(G_\alpha)$, {\sf vertex set} if $\alpha[i]\subseteq V(G_\alpha)$, and {\sf  edge set} otherwise.\footnote{Note that we distinguish between a set containing a single vertex (or edge) and a single vertex (or edge).}
\end{definition}

\begin{definition}{\rm [{\bf Property}]}
A {\em property} is a function $\sigma$ from the set of all structures to $\{\mbox{\sf true, false}\}.$ 
\end{definition}

Finally, we extend the notion of unbreakability to structures.

\begin{definition}{\rm [\bf $(s,c)$-Unbreakable structure]} Let $\alpha$ be a structure. If $G_\alpha$ is an $(s,c)$-unbreakable graph, then we say that $\alpha$ is an {\em $(s,c)$-unbreakable} structure, and otherwise we say that $\alpha$ is  an {\em $(s,c)$-breakable} structure.
\end{definition}

\subsection{Counting Monadic Second Order Logic}\label{sec:prelimsCMSO}

The syntax of Monadic Second Order Logic (MSO) of graphs includes the logical connectives $\vee,$ $\land,$ $\neg,$ 
$\Leftrightarrow,$ $\Rightarrow,$ variables for vertices, edges, sets of vertices and sets of edges, the quantifiers $\forall$ and $\exists$, which can be applied to these variables, and five binary relations: 
\begin{enumerate}
\item $u\in U$, where $u$ is a vertex variable and $U$ is a vertex set variable; 
\item $d \in D$, where $d$ is an edge variable and $D$ is an edge set variable;
\item $\mathbf{inc}(d,u),$ where $d$ is an edge variable, $u$ is a vertex variable, and the interpretation  is that the edge $d$ is incident to  $u$; 
\item $\mathbf{adj}(u,v),$ where $u$ and $v$ are vertex variables, and the interpretation is that $u$ and $v$ are adjacent;
\item  equality of variables representing vertices, edges, vertex sets and edge sets.
\end{enumerate}

Counting Monadic Second Order Logic (CMSO) extends MSO by including atomic sentences testing whether the cardinality of a set is equal to $q$ modulo $r,$ where $q$ and $r$ are integers such that $ 0\leq q<r $ and $r\geq 2$. That is, CMSO is MSO with the following atomic sentence: 
$\mathbf{card}_{q,r}(S) = \mathbf{true}$ if and only if $|S| \equiv q \pmod r$, where $S$ is a set.
We refer to~\cite{ArnborgLS91,Courcelle90,Courcelle97} for a detailed introduction to  CMSO. 

\bigskip
{\noindent\bf Evaluation.}
To evaluate a CMSO-formula $\psi$ on a structure $\alpha$, we instantiate the free variables of $\psi$ by the elements of $\alpha.$ In order to determine which of the free variables of $\psi$ are instantiated by which of the elements of $\alpha$, we introduce the following conventions. First, each free variable $x$ of a CMSO-formula $\psi$ is associated with a rank, $r_x\in\mathbb{N}\setminus\{1\}$. Thus, a CMSO-formula $\psi$ can be viewed as a string accompanied by a tuple of integers, where the tuple consists of one integer $r_x$ for each free variable $x$ of $\psi$. 

Given a structure $\alpha$ and a CMSO-formula $\psi$, we say that {\bf type}$(\alpha)$ {\em matches} $\psi$ if {\bf (i)} the arity of $\alpha$ is  at least  $\max r_x$, where the maximum is taken  over  each  free variable $x$ of $\psi$, and {\bf (ii)} for each free variable $x$ of $\psi$,  ${\bf type}(\alpha)[r_x]$ is compatible with the type of $x.$
For example, if $x$ is a vertex set variable, then ${\bf type}(\alpha)[r_x]={\sf vertex\ set}.$ 
Finally, we say that $\alpha$ {\em matches} $\psi$ if ${\bf type}(\alpha)$ matches $\psi$.
Given a free variable $x$ of a CMSO sentence $\psi$ and a structure $\alpha$ that matches $\psi$, the element corresponding to $x$ in $\alpha$ is $\alpha[r_x]$.
 
\begin{definition}{{\rm [{\bf Property $\sigma_{\psi}$}]}}
Given a CMSO-formula $\psi$, the property $\sigma_{\psi}$ is defined as follows. Given a structure $\alpha$, if $\alpha$ does not match $\psi$, then $\sigma_{\psi}(\alpha)$ equals {\sf false}, and otherwise $\sigma_{\psi}(\alpha)$ equals the result of  the evaluation of $\psi$ where each free variable $x$ of $\psi$ is instantiated by $\alpha[r_{x}].$
\end{definition}
 
Note that some elements of $\alpha$ might not correspond to any variable of $\psi$. However, $\psi$ may still be 
evaluated on the structure $\alpha$---in this case, the evaluation of $\psi$ does not depend on all the elements of the structure. If the arity of $\alpha$ is 1, then we use $\sigma_{\psi}(G_\alpha)$ as a shorthand for $\sigma_{\psi}(\alpha)$.

\begin{definition}{{\rm [{\bf CMSO-definable property}]}}
A property $\sigma$ is {\em CMSO-definable} if there exists a CMSO-formula $\psi$ such that $\sigma=\sigma_{\psi}$.
In this case, we say that $\psi$ {\em defines} $\sigma$.
\end{definition}

\subsection{Boundaried Structures}\label{sec:prelimsBoundariedStruct}

The notion of a {\em boundaried structure} is an extension of the notion of a {\em boundaried graph} and is defined as follows. 

\begin{definition}{\rm [{\bf Boundaried structure]}}
A {\em boundaried structure} is a tuple whose first element is a boundaried graph $G$, denoted by $G_\alpha$, and each of the remaining elements is a subset of $V(G),$ a subset of $E(G),$ a vertex in $V(G),$ an edge in $E(G),$ or the symbol $\star.$ The number of elements in the tuple is the {\em arity} of the boundaried structure.
\end{definition}

Given a boundaried structure $\alpha$ of arity $p$ and an integer $i\in [p]$, we let $\alpha[i]$ denote the $i$'th element of $\alpha$. We remark that we extend the definition of $(s,c)$-(un)breakability of structures, to boundaried structures. Next, other terms presented in previous subsections are adapted to fit boundaried structures.

\begin{definition}{\rm [{\bf Type]}}
Let $\alpha$ be a boundaried structure of arity $p$. The {\em type} of $\alpha$ is a tuple of arity $p$, denoted by ${\bf type}(\alpha)$, where the first element, ${\bf type}(\alpha)[1]$, is {\sf boundaried graph}, 
and for every $i\in\{2,3,\ldots,p\}$, ${\bf type}(\alpha)[i]$ is {\sf vertex} if $\alpha[i]\in V(G_\alpha)$, {\sf edge} if $\alpha[i]\in E(G_\alpha)$, {\sf vertex set} if $\alpha[i]\subseteq V(G_\alpha)$,  {\sf edge set} if $\alpha[i]\subseteq E(G_\alpha)$ and $\star$ otherwise.
\end{definition}

Now, given a boundaried structure and a CMSO-formula $\psi,$ we say that ${\bf type}(\alpha)$ {\em matches} $\psi$
if {\bf (i)} the arity of $\alpha$ is  at least $\max r_x$, where the maximum is taken over each free variable $x$ of $\psi$,
and {\bf (ii)} for each free variable $x$ of $\psi$, ${\bf type}(\alpha)[r_x]$ is compatible with the type of $x.$ Moreover, we say that $\alpha$ matches $\psi$ if ${\bf type}(\alpha)$ {\em matches} $\psi$.

Given $p\in\mathbb{N}$, ${\cal A}^{p}$ denotes the class of all boundaried structures of arity $p$, and given a finite set $I\subseteq \mathbb{N}$, ${\cal A}^{p}_{I}$ (${\cal A}^{p}_{\subseteq I}$) denotes the class of all boundaried  structures of arity $p$ whose boundaried graph belongs to ${\cal F}_{I}$ (resp.~${\cal F}_{\subseteq I}$). A boundaried structure in ${\cal A}^{p}_{\subseteq [t]}$ is called a $t$-{\em boundaried structure}. Finally, we let ${\cal A}$ denote the class of all boundaried structures.

\begin{definition}{\rm [{\bf Compatiblity}] }
Two boundaried structures $\alpha$ and $\beta$ are {\em compatible} (notationally, $\alpha{\sim_c}\beta$) if the following conditions are satisfied.
\begin{itemize}
 \item $\alpha$ and $\beta$ have the same arity $p.$

 \item For every $i\in[p]$:
	\begin{itemize}
	\item ${\bf type}(\alpha)[i]={\bf type}(\beta)[i]\neq \star$, or
	\item ${\bf type}(\alpha)[i]\in\{${\sf vertex},{\sf edge}$\}$ and ${\bf type}(\beta)[i]=\star$, or
	\item ${\bf type}(\beta)[i]\in\{${\sf vertex},{\sf edge}$\}$ and ${\bf type}(\alpha)[i]=\star$.
	\end{itemize}
 
 \item For every $i\in[p]$ such that both $\alpha[i]$ and $\beta[i]$ are vertices: $\alpha[i] \in \delta(G_\alpha),$ $\beta[i] \in \delta(G_\beta)$ and $\lambda_{G_\alpha}(\alpha[i]) = \lambda_{G_\beta}(\beta[i]).$
 
 \item For every $i\in[p]$ such that both $\alpha[i]$ and $\beta[i]$ are edges: $\alpha[i] \in E(G_\alpha[\delta(G_\alpha)]),$ $\beta[i] \in E(G_\beta[\delta(G_\beta)])$ and $\{\lambda_{G_\alpha}(x_{\alpha[i]}),\lambda_{G_\alpha}(y_{\alpha[i]}) \} = \{\lambda_{G_\beta}(x_{\beta[i]}), \lambda_{G_\beta}(y_{\beta[i]})\}$, where $\alpha[i]=(x_{\alpha[i]},y_{\alpha[i]})$ and $\beta[i]=(x_{\beta[i]},y_{\beta[i]})$. That is, $x_{j}$ and $y_{j}$ are the endpoints of the edge $j\in \{\alpha[i],\beta[i]\}$.
 \end{itemize}
 \end{definition}
 
\begin{definition}{\rm [{\bf Gluing by $\oplus$}] }
Given two compatible boundaried structures $\alpha$ and $\beta$ of arity $p$, the operation $\alpha \oplus \beta$ is defined as follows. 
\begin{itemize}
\item $\alpha \oplus \beta$ is a structure $\gamma$ of arity $p$.
\item $G_\gamma = G_\alpha \oplus G_\beta.$ 
\item For every $i\in[p]$:
	\begin{itemize}
	\item if $\alpha[i]$ and $\beta[i]$ are sets, $\gamma[i]=\alpha[i]\cup \beta[i]$;
	\item if $\alpha[i]$ and $\beta[i]$ are vertices/edges, $\gamma[i] = \alpha[i] = \beta[i]$;
	\item if $\alpha[i]=\star$, $\gamma[i] = \beta[i]$;
	\item if $\beta[i]=\star$, $\gamma[i] = \alpha[i].$
	\end{itemize}
\end{itemize}
\end{definition}

\subsection{Finite State}\label{sec:prelimsFiniteState}

This subsection states a variant of the classical Courcelle's Theorem~\cite{Courcelle90,Courcelle92a,Courcelle97} (see also \cite{Courcelle:2012book}), which is a central component in the proof of our main result. To this end, we first define the {\em compatibility equivalence relation $\equiv_c$} on boundaried structures as follows. We say that $\alpha \equiv_c \beta$ if $\Lambda(G_\alpha)=\Lambda(G_\beta)$ and for every boundaried structure $\gamma,$ $$\alpha\sim_c \gamma\iff \beta\sim_c\gamma.$$ Now, we define the {\em canonical equivalence relation} $\equiv_{\sigma}$ on boundaried structures.

\begin{definition}{\rm [\bf Canonical equivalence]}
Given a property $\sigma$ of structures, the {\em canonical equivalence relation} $\equiv_{\sigma}$ on boundaried structures is defined as follows. For two boundaried structures $\alpha$ and $\beta$, we say that $\alpha \equiv_{\sigma} \beta$ if {\bf (i)} $\alpha \equiv_{c} \beta$, and {\bf (ii)} for all boundaried structures $\gamma$ compatible with $\alpha$ (and thus also with $\beta$), we have 
\begin{eqnarray*}
\sigma(\alpha \oplus \gamma)={\sf true} \Leftrightarrow \sigma(\beta \oplus \gamma)={\sf true}. 
\end{eqnarray*}
\end{definition}

It is easy to verify that $\equiv_{\sigma}$ is indeed an equivalence relation. Given a property $\sigma$ of structures, $p\in \mathbb{N}$ and $I\subseteq \mathbb{N}$, we let ${\cal E}_{\equiv_{\sigma}}[{\cal A}^p_{\subseteq I}]$ denote the set of equivalence classes of $\equiv_{\sigma}$ when restricted to ${\cal A}^{p}_{\subseteq I}$.

\begin{definition}{\rm [\bf Finite state]}
A property $\sigma$ of structures is {\em finite state} if, for every $p\in \mathbb{N}$ and $I\subseteq \mathbb{N}$, ${\cal E}_{\equiv_{\sigma}}[{\cal A}^p_{\subseteq I}]$ is finite. 
\end{definition}

Given a CMSO sentence $\psi,$ the canonical equivalence relation associated with $\psi$ is $\equiv_{\sigma_{\psi}}$, and for the sake of simplicity, we denote this relation by $\equiv_{\psi}$.

We are now ready to state the variant of Courcelle's Theorem which was proven in \cite{DBLP:journals/jacm/BodlaenderFLPST16} (see also \cite{Courcelle90,Courcelle92a,Courcelle97}) and which we use in this paper.

\begin{lemma}[\cite{DBLP:journals/jacm/BodlaenderFLPST16}]\label{lem:finiteState}
Every CMSO-definable property on structures has finite state.
\end{lemma}

\subsection{Parameterized Complexity}\label{sec:prelimsPC}

An instance of a parameterized problem is a pair of the form $(x,k),$ where $k$ is a non-negative integer called the {\em parameter}. 
Thus, a parameterized problem $\Pi$ is a subset of $\Sigma^{*}\times \mathbb{N}_0$, for some finite alphabet $\Sigma$.

Two central notions in parameterized complexity are those of {\em uniform fixed-parameter tractability} and {\em non-uniform fixed-parameter tractability}. In this paper, we are interested in the second notion, which is defined as follows.

\begin{definition}{\rm [\bf Non-uniform fixed-parameter tractability ({\FPT})]}\label{def:FPT}
Let $\Pi$ be a parameterized problem. We say that $\Pi$ is {\em non-uniformly fixed-parameter tractable ({\FPT})} if there exists a fixed $d$ such that for every fixed $k\in\mathbb{N}_0$, there exists an algorithm \alg{A}$_k$ that for every $x\in \Sigma^{*}$,  determines whether $(x,k)\in\Pi$ in time $\OO(|x|^d)$.
\end{definition}

Note that in Definition \ref{def:FPT}, $d$ is independent of $k$. We refer to books such as~\cite{DowneyFbook13,DBLP:books/sp/CyganFKLMPPS15} for a detailed introduction to parameterized complexity.

\section{CMSO Model Checking}\label{sec:model}

Given a CMSO formula $\psi$, the {\sc CMSO}$[\psi]$ problem is defined as follows. The input of {\sc CMSO}$[\psi]$ is a structure $\alpha$ that matches $\psi$, and the objective is to output $\sigma_\psi(\alpha)$. In this section, we prove the following result, which then implies Theorem~\ref{thm:main_graphs}.

\begin{restatable}{theorem}{maintheorem}
\label{thm:main}
Let $\psi$ be a CMSO formula. For all $c\in\mathbb{N}$, there exists $s\in\mathbb{N}$ such that if there exists an algorithm that solves {\sc CMSO}$[\psi]$ on $(s,c)$-unbreakable structures in time $\OO(n^d)$ for some $d>4$, then there exists an algorithm that solves {\sc CMSO}$[\psi]$ on general structures in time $\OO(n^{d})$.
\end{restatable}



In the context of parameterized complexity, {\sc min-CMSO}$[\psi]$ ({\sc min-Edge-CMSO}$[\psi]$) is defined as follows. The input of {\sc min-CMSO}$[\psi]$ is a structure $\alpha$, where for all $S\subseteq V(G_\alpha)$ (resp.~$S\subseteq E(G_\alpha)$), $\alpha\diamond S$ matches $\psi$, and a parameter $k$. The objective is to determine whether there exists $S\subseteq V(G_\alpha)$ (resp.~$S\subseteq E(G_\alpha)$) of size at most $k$ such that $\sigma_\psi(\alpha\diamond S)$ is {\sf true}. Similarly, we define {\sc max-CMSO}$[\psi]$ (resp.~{\sc max-Edge-CMSO}$[\psi]$), where the size of $S$ should be at least $k$, and {\sc eq-CMSO}$[\psi]$ (resp.~{\sc eq-Edge-CMSO}$[\psi]$), where the size of $S$ should be exactly $k$. Then, as a consequence of  Theorem \ref{thm:main}, we derive the following result.

\begin{theorem}\label{cor:main}
Let {\sc x}$\in\{${\sc min,max,eq,min-Edge,max-Edge,eq-Edge}$\}$, and let $\widehat{\psi}$ be a CMSO sentence. For all $\widehat{c}:\mathbb{N}_0\rightarrow\mathbb{N}_0$, there exists $\widehat{s}:\mathbb{N}_0\rightarrow\mathbb{N}_0$ such that if {\sc x-CMSO}$[\widehat{\psi}]$ parameterized by $k$ is {\FPT} on $(\widehat{s}(k),\widehat{c}(k))$-unbreakable structures, then {\sc x-CMSO}$[\widehat{\psi}]$ parameterized by $k$ is {\FPT} on general structures.
\end{theorem} 

\begin{proof}
Denote $V=V(G_\alpha)$ and $E=E(G_\alpha)$. First, notice that for every fixed  $k$, there exists a CMSO sentence $\psi_k$ such that {\sc x-CMSO}$[\widehat{\psi}]$ is essentially equivalent to {\sc CMSO}$[\psi_k]$. Indeed, if {\sc x}$\in\{${\sc min,max,eq}$\}$, we can define $\psi_k$ as follows:
\begin{itemize}
\item if {\sc x}={\sc min}, then set $\psi_k=\exists_{S\subseteq V}[(|S|\leq k)\wedge \widehat{\psi}(S)]$;
\item if {\sc x}={\sc max}, then set $\psi_k=\exists_{S\subseteq V}[(|S|\geq k)\wedge \widehat{\psi}(S)]$;
\item if {\sc x}={\sc eq}, then set $\psi_k=\exists_{S\subseteq V}[(|S|=k)\wedge \widehat{\psi}(S)]$.
\end{itemize}
Here, we have that
\begin{itemize}
\item $|S|\leq k$ is the CMSO sentence $\exists_{v_1,\ldots,v_k\in V}[\forall_{u\in V}(u=v_1\vee\cdots\vee u=v_k\vee \neg u\in S)$,
\item $|S|\geq k$ is the CMSO sentence $\exists_{v_1,\ldots,v_k\in V}[v_1\in S\wedge\cdots\wedge v_k\in S\wedge {\bf distinct}(v_1,\ldots,v_k)]$, where ${\bf distinct}(v_1,\ldots,v_k)$ is the CMSO sentence $[(\neg v_1=v_2)\wedge\cdots\wedge(\neg v_1=v_k)]\wedge\cdots\wedge[(\neg v_i=v_1)\wedge\cdots\wedge(\neg v_i=v_{i-1})\wedge(\neg v_i=v_{i+1})\wedge\cdots\wedge(\neg v_i=v_k)]\wedge\cdots\wedge[(\neg v_k=v_1)\wedge\cdots\wedge(\neg v_k=v_{k-1})]$, and
\item $|S|=k$ is the CMSO sentence $(|S|\leq k)\wedge(|S|\geq k)$.
\end{itemize}
In case {\sc x}$\in\{${\sc min-Edge,max-Edge,eq-Edge}$\}$, we replace each occurrence of $V$ by an occurrence of $E$.

Let $\widehat{c}:\mathbb{N}_0\rightarrow\mathbb{N}_0$. Accordingly, define $\widehat{s}:\mathbb{N}_0\rightarrow\mathbb{N}_0$ as follows. For all $k\in\mathbb{N}_0$, let $\widehat{s}(k)$ be the constant $s$ in Theorem \ref{thm:main} where $\psi=\psi_k$ and $c=\widehat{c}(k)$. Suppose that {\sc x-CMSO}$[\widehat{\psi}]$ parameterized by $k$ is {\FPT} on $(s(k),c(k))$-unbreakable structures. Then, there exists a fixed $d>4$ such that for every fixed $k\in\mathbb{N}_0$, there exists an algorithm \alg{A}$_k$ that solves {\sc x-CMSO}$[\widehat{\psi}]$ on $(s(k),c(k))$-unbreakable structures in time $\OO(n^d)$. Thus, we can employ \alg{A}$_k$ to solve {\sc CMSO}$[\psi_k]$ on $(s(k),c(k))$-unbreakable structures in\ time $\OO(n^d)$. By Theorem \ref{thm:main}, we obtain that for every fixed $k\in\mathbb{N}_0$, there exists an algorithm that solves {\sc CMSO}$[\psi_k]$ on general structures in time $\OO(n^{d})$, which implies that for every fixed $k\in\mathbb{N}_0$, there exists an algorithm that solves {\sc x-CMSO}$[\psi]$ on general structures in time $\OO(n^{d})$. We thus conclude that {\sc x-CMSO}$[\widehat{\psi}]$ parameterized by $k$ is {\FPT} on general structures.
\end{proof}

From now on, to prove Theorem~\ref{thm:main}, we assume a fixed CMSO formula $\psi$ and a fixed $c\in\mathbb{N}$. Moreover, we fix $p$ as the number of free variables of $\psi$, and $I=[2c]$.
 We also let $s\in\mathbb{N}$ be fixed, where its exact value (that depends only on $\psi$ and $c$) is determined later. Finally, we assume that there exists an algorithm, \alg{Solve-Unbr-ALG}, that solves {\sc CMSO}$[\psi]$ on $(s,c)$-unbreakable structures in time $\OO(n^d)$ for some $d>4$.

\subsection{Understanding the {\sc CMSO}$[\psi]$ Problem}\label{sec:understanding}

To solve {\sc CMSO}$[\psi]$, we consider a generalization of {\sc CMSO}$[\psi]$, called {\sc Understand}$[\psi]$.
The definition of this generalization is based on an  examination of ${\cal E}_{\equiv_{\psi}}[{\cal A}^p_{\subseteq I}]$. Given a boundaried structure $\alpha\in{\cal A}^p_{\subseteq I}$, we let $E_\alpha$ denote the equivalence class in ${\cal E}_{\equiv_{\psi}}[{\cal A}^p_{\subseteq I}]$ that contains $\alpha$.
For every equivalence class $E_q\in {\cal E}_{\equiv_{\psi}}[{\cal A}^p_{\subseteq I}]$, let $\alpha_{E_q}$ denote some boundaried structure in $E_q$ such that there is no boundaried structure $\alpha\in E_q$ where the length of the string encoding $\alpha$ is smaller than the length of the string encoding $\alpha_{E_q}$. Accordingly, denote ${\cal R}_{\equiv_{\psi}}[{\cal A}^p_{\subseteq I}]=\{\alpha_{E_q}: E_q\in {\cal E}_{\equiv_{\psi}}[{\cal A}^p_{\subseteq I}]\}$. These will be the representatives of the equivalence classes induced by ${\equiv_{\psi}}$.
 By Lemma \ref{lem:finiteState}, there is a fixed $r\in\mathbb{N}$ (that depends only on $\psi$ and $c$) such that both $|{\cal R}_{\equiv_{\psi}}[{\cal A}^p_{\subseteq I}]|$ and the length of encoding of any boundaried structure in ${\cal R}_{\equiv_{\psi}}$ are upper bounded by $r$ as well as $c\leq r$. Note that the encoding explicitly lists all vertices and edges. 
By initially choosing $s$ appropriately, we ensure that \boundS.

The {\sc Understand}$[\psi]$ problem is defined as follows. The input 
 is a boundaried structure $\alpha\in{\cal A}^p_{\subseteq I}$ that matches $\psi$, and the objective is to output a boundaried structure $\beta\in{\cal R}_{\equiv_{\psi}}[{\cal A}^p_{\subseteq I}]$ such that $E_\alpha = E_\beta$.

We proceed by showing that to prove Theorem~\ref{thm:main}, it is sufficient to prove that there exists an algorithm that solves {\sc Understand}$[\psi]$ on general boundaried structures in time $\OO(n^{d})$.

\begin{lemma}\label{lem:understandToSolve}
If there exists an algorithm that solves {\sc Understand}$[\psi]$ on general boundaried structures in time $\OO(n^{d})$, then there exists an algorithm that solves {\sc CMSO}$[\psi]$ on general structures in time $\OO(n^{d})$.
\end{lemma}

\begin{proof}
Let $G_\emptyset$ denote the graph satisfying $V(G_\emptyset)=\emptyset$. Suppose that there exists an algorithm, \alg{Understand-ALG}, that solves {\sc Understand}$[\psi]$ on general structures in time $\OO(n^{d})$. Then, given an input for {\sc CMSO}$[\psi]$, which is a structure $\alpha$, our algorithm works as follows. It lets $\alpha'$ be the boundaried structure that is identical to $\alpha$ except that the (not-boundaried) graph $G_{\alpha}$ is replaced by the boundaried graph defined by $G_{\alpha'}$ and $\delta(G_{\alpha'})=\emptyset$. Then, it calls \alg{Understand-ALG} with $\alpha'$ as input to obtain a boundaried structure $\beta'\in{\cal R}_{\equiv_{\psi}}[{\cal A}^p_{\subseteq I}]$. Next, it lets $\gamma$ be the boundaried structure of arity $p$ where $G_\gamma=G_\emptyset$, and for all $i\in[p]$, if ${\bf type}(\alpha)[i]$ is {\sf vertex} or {\sf edge} then $\gamma[i]=\star$, and otherwise $\gamma[i]=\emptyset$. Moreover, it lets $\beta$ be the (not-boundaried) structure $\beta'\oplus \gamma$. Recall that $s\geq 2r 2^c+r$ where $r$ is an upper bound on the length of any boundaried structure in ${\cal R}_{\equiv_{\psi}}$, and therefore $\beta$ is an $(s,c)$-unbreakable structure. Thus, our algorithm can finally call \alg{Solve-Unbr-ALG} (whose existence we have already assumed) with $\beta$ as input, and outputs the answer that this call returns.

Clearly, the algorithm runs in time $\OO(n^{d})$. Let us now prove that it solves {\sc CMSO}$[\psi]$ correctly. By the correctness of \alg{Understand-ALG}, it holds that $E_{\alpha'}=E_{\beta'}$. In particular, this equality implies that $\sigma_\psi(\alpha'\oplus\gamma)=\sigma_\psi(\beta'\oplus\gamma)$. Notice that $\alpha=\alpha'\oplus\gamma$. Hence, $\sigma_\psi(\alpha)=\sigma_\psi(\beta)$. By the correctness of \alg{Solve-Unbr-ALG}, we thus conclude that our algorithm is correct.
\end{proof}

In light of Lemma \ref{lem:understandToSolve}, the rest of this section focuses on the proof of the following result.

\begin{lemma}\label{lem:understandSuffices}
There exists an algorithm that solves {\sc Understand}$[\psi]$ on general boundaried structures in time $\OO(n^{d})$.
\end{lemma}

\subsection{{\sc Understand}$[\psi]$ on Unbreakable Structures}\label{sec:unbreakable}

Recall that \boundS. In this subsection, we show that Algorithm \alg{Solve-Unbr-ALG} can be used as a subroutine in order to efficiently solve {\sc Understand}$[\psi]$ on $(s-r,c)$-unbreakable boundaried structures. 
For this, we follow the method of test sets (see for example, [Section 12.5, \cite{DowneyFbook13}]). The high level idea here is as follows. We first enumerate the relevant subset of the finite set of minimal representatives. In other words, we simply list those minimal representatives which can be glued in a meaningful way to the structure under consideration, call it  $\alpha$. We now observe that gluing each of these representatives to $\alpha$ results in an $(s,c)$-unbreakable structure, which is what we need to call \alg{Solve-Unbr-ALG}.
In this way we solve the instance obtained by gluing $\alpha$ to each minimal representative. 

Now, for every (not necessarily distinct) pair of minimal representatives, we glue them together and do the same. This way, we can identify the  specific minimal representative whose behaviour when glued with {\em every}  minimal representative, precisely resembles that of the structure $\alpha$ when we do the same with $\alpha$. Consequently, we obtain a solution for  {\sc Understand}$[\psi]$. We now formalize this intuition in the following lemma.


\begin{lemma}\label{lem:understandUnbreakable}
There exists an algorithm  \alg{Understand-Unbr-ALG}, that solves {\sc Understand}$[\psi]$, where it is guaranteed that inputs are $(s-r,c)$-unbreakable boundaried structures, in time $\OO(n^{d})$.\footnote{Here, \alg{Understand-Unbr-ALG} is not requested to verify whether the input is indeed an $(s-r,c)$-unbreakable boundaried structure.}
\end{lemma}

\begin{proof}
We design the algorithm \alg{Understand-Unbr-ALG} as follows. Let $\alpha$ be an input, which is an $(s-r,c)$-unbreakable boundaried structure. Moreover, let ${\cal C}=\{\gamma\in{\cal R}_{\equiv_{\psi}}[{\cal A}^p_{\subseteq I}]: \gamma\equiv_c\alpha\}$, and let ${\cal T}$ denote the set of boundaried structures in ${\cal R}_{\equiv_{\psi}}[{\cal A}^p_{\subseteq I}]$ that are compatible with $\alpha$. In the first phase, the algorithm performs the following computation. Notice that for every $\beta\in {\cal T}$, since $|V(G_\beta)|\leq r$, it holds that $\alpha\oplus\beta$ is an $(s,c)$-unbreakable structure. Thus, for every $\beta\in {\cal T}$, \alg{Understand-Unbr-ALG} can call \alg{Solve-Unbr-ALG} with $\alpha\oplus\beta$ as input, and it lets {\sf ans}$(\alpha,\beta)$ denote the result.

In the second phase, the algorithm performs the following computation. Notice that for every $\gamma\in {\cal C}$ and $\beta\in{\cal T}$, since $|V(G_\beta)|,|V(G_\gamma)|\leq r$, it holds that $\gamma\oplus\beta$ is a $(2r,c)$-unbreakable structure. Thus, since \boundS, for all $\beta\in{\cal C}$ and $\gamma\in {\cal T}$, \alg{Understand-Unbr-ALG} can call \alg{Solve-Unbr-ALG} with $\gamma\oplus\beta$ as input, and it lets {\sf ans}$(\gamma,\beta)$ denote the result.

Finally, in the third phase, for every $\beta\in {\cal C}$, the algorithm performs the following computation. It checks whether for every $\gamma\in {\cal T}$ it holds that {\sf ans}$(\alpha,\gamma)=\textsf{ans}(\beta,\gamma)$, and if the answer is positive, then it outputs $\beta$. Since $\alpha\in{\cal A}^p_{\subseteq I}$, there exists $\beta'\in {\cal C}$ such that $E_\alpha=E_{\beta'}$, and therefore, at the latest, when $\beta=\beta'$, the algorithm terminates. Thus, the algorithm is well defined, and it is clear that it runs in time $\OO(n^{d})$.

To conclude that the algorithm is correct, it remains to show that for all $\beta\in {\cal C}\setminus\{\beta'\}$, there exists $\gamma\in {\cal T}$ such that {\sf ans}$(\alpha,\gamma)\neq \textsf{ans}(\beta,\gamma)$, as this would imply that the algorithm necessarily outputs $\beta'$. For this purpose, suppose by way of contradiction that there exists $\beta\in {\cal C}\setminus\{\beta'\}$ such that for all $\gamma\in {\cal T}$ it holds that {\sf ans}$(\alpha,\gamma)=\textsf{ans}(\beta,\gamma)$. We now argue that $E_\beta=E_{\beta'}$ which leads to a contradiction since  each boundaried structure in ${\cal R}_{\equiv_{\psi}}[{\cal A}^p_{\subseteq I}]$ belongs to a different equivalence class.

For all $\gamma\in {\cal T}$, since  it holds that {\sf ans}$(\alpha,\gamma)=\textsf{ans}(\beta,\gamma)$,  it also holds that {\sf ans}$(\beta',\gamma)=\textsf{ans}(\beta,\gamma)$.  This implies that $\sigma_\psi(\beta'\oplus\gamma)=\sigma_\psi(\beta\oplus\gamma)$. Consider some boundaried structure $\gamma$ (not necessarily in $\cal T$) that is compatible with $\beta'$ (and thus also with $\beta$). We claim that $\sigma_\psi(\beta'\oplus\gamma)=\sigma_\psi(\beta\oplus\gamma)$. Indeed, let $\gamma'$ be the (unique) boundaried structure in ${\cal R}_{\equiv_{\psi}}[{\cal A}^p_{\subseteq I}]$ such that $E_{\gamma'}=E_\gamma$. Then, $\sigma_\psi(\beta'\oplus\gamma')=\sigma_\psi(\beta'\oplus\gamma)$ and $\sigma_\psi(\beta\oplus\gamma')=\sigma_\psi(\beta\oplus\gamma)$. Note that since $\gamma'$ is compatible with $\beta'$, it is also compatible with $\alpha$, and hence $\gamma'\in{\cal T}$. Therefore, $\sigma_\psi(\beta'\oplus\gamma')=\sigma_\psi(\beta\oplus\gamma')$. Overall, we obtain that indeed $\sigma_\psi(\beta'\oplus\gamma)=\sigma_\psi(\beta\oplus\gamma)$.

Note that $\beta\equiv_c\beta'$, and thus, since we have shown that for every boundaried structure $\gamma$ compatible with $\beta'$ it holds that $\sigma_\psi(\beta'\oplus\gamma)=\sigma_\psi(\beta\oplus\gamma)$, we derive that $E_\beta=E_{\beta'}$. However, each boundaried structure in ${\cal R}_{\equiv_{\psi}}[{\cal A}^p_{\subseteq I}]$ belongs to a different equivalence class, and thus we have reached the desired  contradiction.
\end{proof}

\subsection{{\sc Understand}$[\psi]$ on General Structures}\label{sec:breakable}

\noindent{\bf The Algorithm \alg{Understand-ALG}.} We start by describing an algorithm called \alg{Understand-ALG}, which is based on recursion. Given an input to {\sc Understand}$[\psi]$ on general boundaried structures, which is a boundaried structure $\alpha$, the algorithm works as follows. First, it calls \alg{Break-ALG} (given by Lemma \ref{lem:break}) with $G_\alpha$ as input to either obtain an $\displaystyle{(\frac{s-r}{2^c},c)}$-witnessing separation $(X,Y)$ or correctly conclude that $G_\alpha$ is $(s-r,c)$-unbreakable. In the second case or if $n<2(s-r)$, it calls \alg{Understand-Unbr-ALG} (given by Lemma \ref{lem:understandUnbreakable}), and returns its output. Next, suppose that \alg{Understand-ALG} obtained an $\displaystyle{(\frac{s-r}{2^c},c)}$-witnessing separation $(X,Y)$ and that $n\geq 2(s-r)$. Without loss of generality, assume that $|X\cap\delta(G_\alpha)|\leq|Y\cap\delta(G_\alpha)|$. Denote $\Delta=\{v\in X\cap Y: v\notin \delta(G_\alpha)\}$.

Now, we define a boundaried structure, $\beta\in{\cal A}^p_{\subseteq I}$, which can serve as an instance of {\sc Understand}$[\psi]$. First, we let the graph $G_\beta$ be $G_\alpha[X]$, and we define $\delta(G_\beta)=(X\cap\delta(G_\alpha))\cup\Delta$. Now, for all $v\in X\cap\delta(G_\alpha)$, we define $\lambda_{G_\beta}(v)=\lambda_{G_\alpha}(v)$. Since $|X\cap\delta(G_\alpha)|\leq|Y\cap\delta(G_\alpha)|$, $\alpha\in{\cal A}^p_{\subseteq I}$ and $|X\cap Y|\leq c$, we have that $|(X\cap\delta(G_\alpha))\cup\Delta|\leq 2c$. Thus, to each $v\in \Delta$, we can let $\lambda_{G_\beta}(v)$ assign some unique integer from $I\setminus\lambda_{G_\alpha}(X\cap\delta(G_\alpha))$. Hence, $G_\beta\in{\cal F}_{\subseteq I}$. Now, for all $i\in\{2,\ldots,p\}$, we set $\beta[i]$ as~follows.
\begin{itemize}
\item If ${\bf type}(\alpha)[i]\in\{${\sf vertex},{\sf edge}$\}$: If $\alpha[i]\in V(G_\beta)\cup E(G_\beta)$, then $\beta[i]=\alpha[i]$, and otherwise $\beta[i]=\star$.
\item Else: $\beta[i]=\alpha[i]\cap(V(G_\beta)\cup E(G_\beta))$.
\end{itemize}

\alg{Understand-ALG} proceeds by calling itself recursively with $\beta$ as input, and it lets $\beta'$ be the output of this call. 

Now, we define another boundaried structure, $\gamma\in{\cal A}^p_{\subseteq I}$, which can serve as an instance of {\sc Understand}$[\psi]$. First, we define the boundaried graph $G_\gamma$ as follows. Let $H$ be the disjoint union of $G_{\beta'}$ and $G[Y]$, where both $G_{\beta'}$ and $G[Y]$ are treated as not-boundaried graphs. For all $v\in X\cap Y$, identify (in $H$) the vertex $v$ of $G[Y]$ with the vertex $u$ of $G_{\beta'}$ that satisfies $\lambda_{G_{\beta'}}(u)=\lambda_{G_\beta}(v)$, and for the sake of simplicity, let $v$ and $u$ also denote the identity of the resulting (unified) vertex. The graph $G_\gamma$ is the result of this process. Moreover, let $\Delta'$ denote the set of vertices in $G_{\beta'}$ whose labels belong to $G_{\beta}(\Delta).$
Next, set $\delta(G_\gamma)=(Y\cap\delta(G_\alpha))\cup(\delta(G_{\beta'})\setminus\Delta')$.
Now, for all $v\in Y\cap\delta(G_\alpha)$, we define $\lambda_{G_\gamma}(v)=\lambda_{G_\alpha}(v)$, and for all $v\in\delta(G_{\beta'})\setminus\Delta'$, we define $\lambda_{G_\gamma}(v)=\lambda_{G_{\beta'}}(v)$ (note that if a vertex belongs to both $Y\cap\delta(G_\alpha)$ and $\delta(G_{\beta'})\setminus\Delta'$, we still assign it the same label). Hence, $G_\gamma\in{\cal F}_{\subseteq I}$.  For the sake of simplicity, if two vertices have the same label (one in $G_\alpha$ and the other in $G_\gamma$), we let the identity of one of them also refer to the other and vice versa. For all $i\in\{2,\ldots,p\}$, we set $\gamma[i]$ to have the same type as $\alpha[i]$, and define it as~follows.
\begin{itemize}
\item If ${\bf type}(\alpha)[i]\in\{${\sf vertex},{\sf edge}$\}$: If $\alpha[i]\in V(G_\gamma)\cup E(G_\gamma)$, then $\gamma[i]=\alpha[i]$, and otherwise $\gamma[i]=\star$.
\item Else: $\gamma[i]=\alpha[i]\cap(V(G_\gamma)\cup E(G_\gamma))$.
\end{itemize}

Finally, \alg{Understand-ALG} calls itself recursively with $\gamma$ as input, and it returns $\gamma'$, the output of this call.

\bigskip
\noindent{\bf Correctness.} Here, we prove the following result.
\begin{lemma}\label{lem:correctUnderstand-ALG}
If \alg{Understand-ALG} terminates, then it correctly solves {\sc Understand}$[\psi]$ on general boundaried structures.
\end{lemma}

\begin{proof}
The proof is by induction on the number of recursive calls that the algorithm performs. Here, we suppose that at a given call which terminates, the recursive calls (which must then also terminate) return correct answers. At the basis, we are at a call where the algorithm performs no recursive calls. Thus, the basis corresponds to calls where either $n < 2(s-r)$ or \alg{Break-ALG} concludes that $G_\alpha$ is $(s-r,c)$-unbreakable; then, correctness follows from Lemma \ref{lem:understandUnbreakable}. Next, consider a call that terminates and where the algorithm calls itself recursively.

We need to show that $E_{\alpha}=E_{\gamma'}$. Since we assume that the recursive calls return correct answers, it is sufficient to show that $E_\alpha=E_\gamma$. Moreover, due to this assumption, it also holds that $E_{\beta'}=E_\beta$.

First, we show that $\alpha\equiv_c\gamma$. By the definition of $G_\gamma$, every vertex in $Y$ that has a label in $G_\alpha$, is present in $G_\gamma$ and has the same label in $G_\gamma$. Moreover, every vertex $v$ in $X\setminus Y$ that has a label $\ell$ in $G_\alpha$, also has the same label in $G_\beta$ (by the definition of $G_\beta$). Therefore, this label is also present in $G_{\beta}'$ (since $E_{\beta'}=E_\beta$) and not given to a vertex in $\Delta'$ (since it is not given to a vertex in $\Delta$). Thus, this label is also present in $G_\gamma$ (by the definition of $G_\gamma$)---in this context, recall that we refer to the vertex in $G_\gamma$ that has the label $\ell$ by $v$ as well. Thus, we have that $\Delta(G_\alpha)=\Delta(G_\gamma)$. Moreover, since $\alpha$ defines an input instance, its arity is $p$, which is also the arity of $\gamma$ (by the definition of $\gamma$). The definition of $\gamma$ also immediately implies that for all $i\in[p]$, ${\bf type}(\alpha)[i]={\bf type}(\gamma)[i]$, and by the above arguments, it also implies that if both $\alpha[i]$ and $\gamma[i]$ are vertices/edges, then $\alpha[i]=\gamma[i]$ and this vertex has the same label in both $G_\alpha$ and $G_\gamma$. This concludes the proof that $\alpha\equiv_c\gamma$.

Now, to derive that $E_\alpha=E_\gamma$, we also need to show that given any boundaried structure $\eta$ compatible with $\alpha$ (and thus also with $\gamma$), it holds that $\sigma_\psi(\alpha \oplus \eta)={\sf true}$ if and only if $\sigma_\psi(\beta \oplus \eta)={\sf true}$. For this purpose, consider some boundaried structure $\eta$ compatible with $\alpha$. First, we define a boundaried structure $\mu$ in ${\cal A}^p_{\subseteq I}$ such that $\beta\oplus\mu = \alpha\oplus\eta$ as follows. We let the graph $G_\mu$ be $G_{\alpha\oplus \eta}\setminus(V(G_\beta)\setminus \delta(G_\beta))$, and we define $\delta(G_\mu)=\delta(G_\beta)$. Now, for all $v\in \delta(G_\mu)$, we define $\lambda_{G_\mu}(v)=\lambda_{G_\beta}(v)$. For all $i\in\{2,\ldots,p\}$, we set $\mu[i]$ as~follows.
\begin{itemize}
\item If ${\bf type}(\beta)[i]\in\{${\sf vertex},{\sf edge}$\}$: If $\beta[i]\in V(G_\mu)\cup E(G_\mu)$, then $\mu[i]=\beta[i]$, and otherwise $\mu[i]=\star$.
\item Else: $\mu[i]=\beta[i]\cap(V(G_\mu)\cup E(G_\mu))$.
\end{itemize}

Second, we define a boundaried structure $\rho$ in ${\cal A}^p_{\subseteq I}$ such that $\beta'\oplus\rho = \gamma\oplus\eta$ as follows. We let the graph $G_\rho$ be $G_{\gamma\oplus \eta}\setminus (V(G_{\beta'})\setminus \delta(G_{\beta'}))$, and we define $\delta(G_\rho)=\delta(G_{\beta'})$. Now, for all $v\in \delta(G_\rho)$, we define $\lambda_{G_\rho}(v)=\lambda_{G_{\beta'}}(v)$. For all $i\in\{2,\ldots,p\}$, we set $\rho[i]$ as~follows.
\begin{itemize}
\item If ${\bf type}({\beta'})[i]\in\{${\sf vertex},{\sf edge}$\}$: If ${\beta'}[i]\in V(G_\rho)\cup E(G_\rho)$, then $\rho[i]={\beta'}[i]$, and otherwise $\rho[i]=\star$.
\item Else: $\rho[i]={\beta'}[i]\cap(V(G_\rho)\cup E(G_\rho))$.
\end{itemize}

However, by our definition of $\gamma$, we have that $\mu=\rho$. Indeed, since $E_\beta=E_{\beta'}$ and as we reuse vertex identities (namely, we treat equally labeled vertices in $G_\beta$ and $G_{\beta'}$ as the same vertex), we have that for all $i\in\{2\ldots,p\}$, it holds that $\beta[i]=\beta'[i]$. Thus, to derive that $\mu=\rho$, it is sufficient to show that $G_\mu=G_\rho$, that is, $G_{\alpha\oplus \eta}\setminus(V(G_\beta)\setminus \delta(G_\beta))=G_{\gamma\oplus \eta}\setminus (V(G_{\beta'})\setminus \delta(G_{\beta'}))$. The correctness of this claim follows by noting that $\delta(G_\beta)=\delta(G_{\beta'})$, and thus $G_{\alpha}\setminus(V(G_\beta)\setminus \delta(G_\beta))=G_\alpha[Y\cup\delta(G_\beta)]=G_\gamma[Y\cup\delta(G_{\beta'})]=G_{\gamma}\setminus (V(G_{\beta'})\setminus \delta(G_{\beta'}))$ (by the definition of $G_\gamma$). 

Finally, since $E_\beta=E_{\beta'}$, we have that $\sigma_\psi(\beta \oplus \mu)={\sf true}$ if and only if $\sigma_\psi(\beta' \oplus \mu)={\sf true}$. As $\beta \oplus \mu=\alpha\oplus \eta$ and $\beta'\oplus \rho = \beta'\oplus \mu = \gamma\oplus \eta$, we conclude that $\sigma_\psi(\alpha \oplus \eta)={\sf true}$ if and only if $\sigma_\psi(\beta \oplus \eta)={\sf true}$.
\end{proof}

\bigskip
\noindent{\bf Time Complexity.} Finally, we prove the following result, which together with Lemma \ref{lem:correctUnderstand-ALG}, implies that Lemma \ref{lem:understandSuffices} is correct.

\begin{lemma}\label{lem:timeUnderstand-ALG}
\alg{Understand-ALG} runs in time $\OO(n^{d})$.
\end{lemma}

\begin{proof}
We prove that \alg{Understand-ALG} runs in time bounded by $x\cdot n^{d}$ for some fixed $x$ (to be determined). The proof is by induction on the number of recursive calls that the algorithm performs. At the basis, we are at a call where the algorithm performs no recursive calls. Thus, the basis corresponds to calls where either $n<2(s-r)$ or \alg{Break-ALG} concludes that $G_\alpha$ is $(s-r,c)$-unbreakable; then, by choosing $x$ that is large enough (but independent of the input instance), correctness follows from Lemmata \ref{lem:break} and \ref{lem:understandUnbreakable}. Next, consider a call where the algorithm calls itself recursively.

Denote $n'=|V(G_\beta)|=|X\setminus Y| + |X\cap Y|$ and $\widehat{n}=|V(G_\gamma)|=|Y\setminus X| + |V(G_{\beta'})|$. By Lemma \ref{lem:break} and the inductive hypothesis, there exists a fixed $y$ (independent of $x$ and the input instance) such that \alg{Understand-ALG} runs in time bounded by
\[\begin{array}{l}
y\cdot n^3\log n + x\cdot(n'^{d} + \widehat{n}^{d}).
\end{array}\]

Recall that $c\leq r$. Denote $a=|X\setminus Y|$, $b=|Y\setminus X|$ and $\widehat{s}=s-r$. Then, $n=a+b+|X\cap Y|\leq a+b+c$ and $\widehat{s}/2^c\leq a,b$. Thus, $\widehat{s}/2^c\leq n'=n-\widetilde{n}+c\leq n-\widetilde{n}+r$ and $\widehat{s}/2^c\leq \widehat{n}\leq n-n'+r$. Hence, the running time can further be bounded by
\[\begin{array}{l}
\medskip
y\cdot n^3\log n + x\cdot((\widehat{s}/2^c)^{d} + (n-\widehat{s}/2^c+r)^{d})\leq \\

\smallskip
x\cdot n^{d} + (y\cdot n^3\log n+x\cdot(\widehat{s}/2^c)^{d}+x\cdot rn^{d-1})\\

\hspace{3.6em}- x\cdot (\widehat{s}/2^c)n^{d-1}.
\end{array}\]

Denote $t=\widehat{s}/2^c$. Thus, it remains to show that
\[xtn^{d-1} \geq yn^3\log n+xt^{d}+xrn^{d-1}.\]

Now, recall that \boundS, and therefore $t>2r$. Thus, it is sufficient to show that
\[xtn^{d-1}/2 \geq yn^3\log n+xt^{d}.\]

Since $d>4$, by ensuring that $x/4\geq y$, we further have that it is sufficient to show that 
\[xtn^{d-1}/4 \geq xt^{d}.\]

Finally, recall that $n\geq 2(s-r)$, and therefore $n^{d-1}\geq 4t^{d-1}$. Thus, the inequality above holds.
\end{proof}

\section{Applications}\label{sec:applications}

In this section, we first show how Theorem~\ref{cor:main} can be easily deployed to show the fixed parameter tractability of a wide range of problems of the following kind. The input is a graph $G$ and the task is to find a connected induced subgraph of $G$ of bounded treewidth such that ``few'' vertices outside this subgraph have neighbors inside the subgraph, and additionally the subgraph has a CMSO-definable property. Then, we show that technical problem-specific ingredients of a powerful method for designing parameterized algorithms  called recursive understanding, can be replaced by a black-box invocation of Theorem~\ref{cor:main}. Here, we consider the {\sc Vertex Multiway Cut-Uncut (V-MWCU)} problem as an illustrative example.

\subsection{``Pendant'' Subgraphs with CMSO-Definable Properties}

Formally, given a CMSO sentence $\psi$ and a non-negative integer $t$, the $t$-{\sc Pendant}$[\psi]$ problem is defined as follows. The input of $t$-{\sc Pendant}$[\psi]$ is a graph $G$ and a parameter $k$, and the objective is to determine whether there exists $U\subseteq V(G)$ such that $G[U]$ is a connected graph of treewidth at most $t$, $|N(U)|\leq k$ and $\sigma_{\psi}(G[U])$ is {\sf true}.

We start by defining a CMSO formula $\varphi$ with free variable $S$ as follows.
\[\begin{array}{ll}
\varphi = & \exists_{U\subseteq V(G)}[(G[U]\models\psi)\wedge(G[U]\models {\bf tw}_t)\wedge{\bf conn}(U)\wedge\\
& (\forall_{v\in S}\neg (v\in U))\wedge(\forall_{v\in U}\forall_{u\in V(G)\setminus (U\cup S)}\neg{\bf adj}(v,u))
],
\end{array}\]
where {\bf conn}$(U)$ is the standard CMSO sentence that tests whether $G[U]$ is a connected graph (see, e.g., \cite{DBLP:books/sp/CyganFKLMPPS15}), and {\bf tw}$_t$ is the standard CMSO sentence that tests whether the treewidth of a graph is at most $t$ (see, e.g., \cite{DBLP:conf/focs/FominLMS12}). We remark that  {\bf tw}$_t$ can be constructed by observing that there exists a finite set of graphs, ${\cal M}$, such that a graph has treewidth at most $t$ if and only if it excludes every graph in ${\cal M}$ as a minor, and it is known how to construct a CMSO sentence that tests the exclusion of a fixed graph as a minor (see, e.g., \cite{DBLP:conf/focs/FominLMS12}).

Having defined $\psi$, it is immediate that the $t$-{\sc Pendant}$[\psi]$ problem is equivalent to {\sc min-CMSO}$[\varphi]$ as follows.
\begin{observation}\label{obs:pendantCMSO}
Let $G$ be a graph, and let $k$ be a parameter. Then, $(G,k)$ is a \Yes-instance of $t$-{\sc Pendant}$[\psi]$  if and only if $((G),k)$ is a \Yes-instance of {\sc min-CMSO}$[\varphi]$.
\end{observation}

Next, we solve $t$-{\sc Pendant}$[\psi]$ on unbreakable graphs with the appropriate parameters. Define $c:\mathbb{N}_0\rightarrow\mathbb{N}_0$ as follows. For all $k\in\mathbb{N}_0$, let $\widehat{c}(k)=k+t$. Let $s:\mathbb{N}_0\rightarrow\mathbb{N}_0$ be the function $\widehat{s}$ in Theorem \ref{cor:main} with $\widehat{\psi}=\varphi$ and $\widehat{c}=c$. We first prove the following lemma.

\begin{lemma}\label{lem:pendantSmall}
Let $(G,k)$ be a \Yes-instance of $t$-{\sc Pendant}$[\psi]$ parameterized by $k$ on $(s(k),k+t)$-unbreakable graphs. Then, there exists $U\subseteq V(G)$ such that $G[U]$ is a connected graph of treewidth at most $t$, $|N(U)|\leq k$, $\sigma_{\psi}(G[U])$ is {\sf true} and $|U|<3(s(k)+t)$.
\end{lemma}

\begin{proof}
Since $(G,k)$ is a \Yes-instance, there exists $U\subseteq V(G)$ such that $G[U]$ is a connected graph of treewidth at most $t$, $|N(U)|\leq k$ and $\sigma_{\psi}(G[U])$ is {\sf true}. Moreover, since the treewidth of $G[U]$ is at most $t$, it is easy to see that there exists a separation $(X,Y)$ of order at most $t$ of $G[U]$ such that $|X|,|Y|\geq |U|/3$ (see, e.g., \cite{DBLP:books/sp/CyganFKLMPPS15}). Then, set $X'=X\cup N(U)$ and $Y'=(V(G)\setminus X)\cup (X\cap Y)$. Note that $(X',Y')$ is a separation of order $|X\cap Y|+|N(U)|\leq k+t$. Moreover, $X\setminus Y\subseteq X'\setminus Y'$ and $Y\setminus X\subseteq Y'\setminus X'$. Thus, $(X',Y')$ is a $(|U|/3-t,k+t)$-witnessing separation. Since $G$ is $(s(k),k+t)$-unbreakable graph, we have that $|U|/3-t<s(k)$. Therefore, $|U|<3(s(k)+t)$, which concludes the correctness of the lemma.
\end{proof}

We also need the following result, proved by Fomin and Villanger \cite{DBLP:journals/combinatorica/FominV12}.

\begin{lemma}[\cite{DBLP:journals/combinatorica/FominV12}]\label{lem:enumConn}
Fix $p,q\in\mathbb{N}_0$. Given a graph $G$ and a vertex $v\in V(G)$, the number of subsets $U\subseteq V(G)$ such that $v\in U$, $G[U]$ is a connected graph, $|U|\leq p$ and $|N(U)|\leq q$ is upper bounded by ${p+q\choose p}$ and they can be enumerated in constant time (dependent only on $p$ and $q$).
\end{lemma}

\begin{lemma}\label{lem:pendantUnbreak}
$t$-{\sc Pendant}$[\psi]$ parameterized by $k$ is {\FPT} on $(s(k),k+t)$-unbreakable graphs.
\end{lemma}

\begin{proof}
Fix some $k\in\mathbb{N}_0$. Given an $(s(k),k+t)$-unbreakable graph $G$, our algorithm \alg{A}$_k$ works as follows. By using the algorithm in Lemma \ref{lem:enumConn}, for every vertex $v\in V(G)$, it first computes (in constant time) the set ${\cal U}_v$ of subsets $U\subseteq V(G)$ such that $v\in U$, $G[U]$ is a connected graph, $|U|\leq 3(s(k)+t)$ and $|N(U)|\leq k$. Then, it sets ${\cal U}=\bigcup_{v\in V(G)}{\cal U}_v$. For each $U\in{\cal U}$, since $|U|\leq 3(s(k)+t)=\OO(1)$, the algorithm can test (in constant time) whether $G[U]\models\psi$ and the treewidth of $G[U]$ is at most $t$. 

By Lemma \ref{lem:enumConn}, it holds that $|{\cal U}|=\OO(n)$, and therefore \alg{A}$_k$ runs in time $\OO(n)$. The correctness of \alg{A}$_k$ directly follows from Lemmata \ref{lem:pendantSmall} and \ref{lem:enumConn}. This concludes the proof of the lemma.
\end{proof}

Finally, by Theorem \ref{cor:main}, Observation \ref{obs:pendantCMSO} and Lemma \ref{lem:pendantUnbreak}, we derive the following result.

\begin{theorem}\label{thm:pendant}
$t$-{\sc Pendant}$[\psi]$ parameterized by $k$ is {\FPT} on general graphs.
\end{theorem}

\subsection{Recursive Understanding as a Black Box}

The {\sc Vertex Multiway Cut-Uncut (V-MWCU)} problem is defined as follows. The input of {\sc V-MWCU} consists of a graph $G$, a terminal set $T\subseteq V(G)$, an equivalence relation $\cal R$ on $T$, and a parameter $k$. The objective is to determine whether there exists a subset $U\subseteq V(G)\setminus T$ such that $|U|\leq k$, and for all $u,v\in T$, it holds that $u$ and $v$ belong to the same connected component of $G\setminus U$ if and only if $(u,v)\in{\cal R}$.
Our goal is to prove the following result.

\begin{theorem}\label{thm:multiway}
{\sc V-MWCU} parameterized by $k$ is {\FPT} on general graphs.
\end{theorem}

For syntactic reasons, we view {\sc V-MWCU} as the {\sc Vertex Red-Blue Cut-Uncut (V-RBCU)} problem, which we define as follows. The input of {\sc V-RBCU} consists of a graph $G$, an edge-set $R\subseteq E(G)$ such that $G[R]$ is a cluster graph, and a parameter $k$. The objective is to determine whether there exists a subset $S\subseteq V(G)\setminus V[R]$, called a {\em solution}, such that $|S|\leq k$, and for every two vertices $u,v\in V[R]$, it holds that $u$ and $v$ belong to the same connected component of $(G\setminus S)\setminus R$ if and only if there exists an edge in $R$ whose endpoints are $u$ and $v$.

Given an instance $(G,T,{\cal R},k)$ of {\sc V-MWCU}, we construct (in polynomial time) an equivalent instance $(G',R,k)$ of {\sc V-RBCU} as follows. We set $V(G')=V(G)$, and initialize $E(G')=E(G)$ and $R=\emptyset$. Then, for each $u,v\in T$ such that $(u,v)\in {\cal R}$, we insert into both $E(G')$ and $R$ a new edge whose endpoints are $u$ and $v$. Thus, to prove Theorem \ref{thm:multiway}, it is sufficient to prove the following result.

\begin{lemma}\label{lem:redBlueMain}
{\sc V-RBCU} parameterized by $k$ is {\FPT} on general graphs.
\end{lemma}

We start by defining a CMSO formula $\varphi$ with free variables $R$ and $S$ as follows.
\[\begin{array}{ll}
\varphi =& {\bf cluster}(R)\wedge[\forall_{v\in S}\neg\exists_{e\in R}\ {\bf inc}(e,v)]\wedge\\

&[\forall_{u,v\in V(G)\setminus S}\ \varphi_1\vee\varphi_2\vee\varphi_3],
\end{array}\]
where 
\[\begin{array}{ll}
\varphi_1 =& \neg\exists_{e\in R}\ {\bf inc}(e,u) \vee\neg\exists_{e\in R}\ {\bf inc}(e,v),\\

\varphi_2 =&  \exists_{U\subseteq V(G)\setminus S}[u\in U\wedge v\in U\wedge{\bf conn}(U)]\wedge\\
 & \exists_{e\in R}({\bf inc}(e,u)\wedge {\bf inc}(e,v)),\\

\varphi_3= & \neg\exists_{U\subseteq V(G)\setminus S}[u\in U\wedge v\in U\wedge{\bf conn}(U)]\wedge\\
 & \neg\exists_{e\in R}({\bf inc}(e,u)\wedge {\bf inc}(e,v)),
\end{array}\]
and {\bf cluster}$(R)$ is the standard CMSO sentence that tests whether $G[R]$ is a cluster graph. For completeness,
\[\begin{array}{l}
{\bf cluster}(R) = \forall_{u,v,w\in V(G)}[\varphi_1\vee \neg\exists_{e\in R}\ {\bf inc}(e,w)\vee\\
\neg\exists_{e\in R}({\bf inc}(e,u)\wedge {\bf inc}(e,v))\vee \neg\exists_{e\in R}({\bf inc}(e,v)\wedge {\bf inc}(e,w))\\
\vee\exists_{e\in R}({\bf inc}(e,u)\wedge {\bf inc}(e,w))].
\end{array}\]

Having defined $\varphi$, it is immediate that {\sc V-RBCU} is equivalent to {\sc min-CMSO}$[\varphi]$ as follows.
\begin{observation}\label{obs:redBlueCMSO}
Let $G$ be a graph, and let $k$ be a parameter. Then, $(G,R,k)$ is a \Yes-instance of {\sc V-RBCU} if and only if $((G,R),k)$ is a \Yes-instance of {\sc min-CMSO}$[\varphi]$.
\end{observation}

Next, we solve {\sc V-RBCU} on unbreakable graphs with the appropriate parameters. Define $c:\mathbb{N}_0\rightarrow\mathbb{N}_0$ as follows. For all $k\in\mathbb{N}_0$, let $\widehat{c}(k)=k$. Let $s:\mathbb{N}_0\rightarrow\mathbb{N}_0$ be the function $\widehat{s}$ in Theorem \ref{cor:main} with $\widehat{\psi}=\varphi$ and $\widehat{c}=c$. Given an instance $(G,R,k)$ of {\sc V-RBCU}, let $R_1,R_2,\ldots,R_r$ denote the vertex sets of the cliques in $G[R]$ for the appropriate $r$. We first prove the following lemma.

\begin{lemma}\label{lem:redBlueSmall}
Let $(G,R,k)$ be a \Yes-instance of {\sc V-RBCU} parameterized by $k$ on $(s(k),k)$-unbreakable graphs. Then, there exists a solution $S$ and $i\in [r]$ such that for all $j\in[r]\setminus\{i\}$, $|V(C_j)|\leq s(k)$, where $C_j$ is the connected component of $(G\setminus S)\setminus R$ whose vertex-set contains $R_j$.
\end{lemma}

\begin{proof}
Since $(G,R,k)$ is a \Yes-instance, there exists a solution $S$. For all $j\in[r]$, let $C_j$ is the connected component of $(G\setminus S)\setminus R$ whose vertex-set contains $R_j$. Let $i$ denote an index in $[r]$ that maximizes $|V(C_j)|$. We claim that for all $j\in[r]\setminus\{i\}$, $|V(C_j)|\leq s(k)$. Suppose, by way of contradiction, that there exists $j\in[r]\setminus\{i\}$ such that $|V(C_j)|>s(k)$. Then, since $V(C_j)\subseteq V(G)\setminus(V(C_i)\cup S)$, we have that $(V(C_i)\cup S,V(G)\setminus V(C_i))$ is an $(s(k),k)$-witnessing separation of $G\setminus R$. Since $S$ is a solution, there is no edge in $R$ with one endpoint in $V(C_i)$ and another endpoint outside $V(C_i)$. Therefore, $(V(C_i)\cup S,V(G)\setminus V(C_i))$ is also an $(s(k),k)$-witnessing separation of $G$, which contradicts the fact that $G$ is an $(s(k),k)$-unbreakable graph.
This concludes the proof of the lemma.
\end{proof}

\begin{lemma}\label{lem:redBlueUnbreak}
{\sc V-RBCU} parameterized by $k$ is {\FPT} on $(s(k),k)$-unbreakable graphs.
\end{lemma}

\begin{proof}
Fix some $k\in\mathbb{N}_0$. Given an instance $(G,R,k)$ of {\sc V-RBCU} where $G$ is an $(s(k),k)$-unbreakable graph, our algorithm, \alg{A}$_k$,  works as follows. For all $j\in[r]$, it selects a vertex $v_j\in R_j$ (arbitrarily).
By using the algorithm in Lemma \ref{lem:enumConn}, for all $j\in[r]$, \alg{A}$_k$ computes (in constant time) the set ${\cal U}_j$ of subsets $U\subseteq V(G)$ such that $v_j\in U$, $G[U]$ is a connected graph, $|U|\leq s(k)$ and $|N(U)|\leq k$. Then, for all $j\in[r]$ and $U\in{\cal U}_j$, if it does not hold  that $R_j\subseteq U$ and $(\bigcup_{\ell\in[r]\setminus\{j\}}R_\ell)\cap U=\emptyset$, \alg{A}$_k$ removes $U$ from ${\cal U}_j$.
Afterwards, for every $i\in[r]$, \alg{A}$_k$ calls the recursive procedure \alg{B}$_k$, whose pseudocode is given below, with $i$ and a set $S$ that is initialized to be $\emptyset$.
\begin{enumerate}
\item If $|S|>k$: Output \No.
\item Else if there exists $j\in[r]$ such that $R_j$ is not a subset of the vertex-set of a single connected component of $G\setminus S$: Output \No.
\item Else if for all distinct $j,t\in[r]$, $R_j$ and $R_t$ are subsets of distinct vertex-sets of connected components of $G\setminus S$: Output \Yes.
\item\label{step:recurse} Else:
	\begin{enumerate}
	\item Let $j$ be an index in $[r]\setminus\{i\}$ for which there exists $t\in[r]\setminus\{j\}$ such that $R_j$ and $R_t$ are subsets of the vertex-set of a single connected component of $G\setminus S$. 
	\item For all $U\in {\cal U}_j$: If \alg{B}$_k(i,S\cup N(U))$ outputs \Yes, then output \Yes.
	\item Return \No.
	\end{enumerate}
\end{enumerate}
If no call outputted \Yes, then \alg{A}$_k$ outputs \No.

Note that at each recursive call, the size of $S$ increases by at least 1. Indeed, we only update $S$ by inserting vertices into it, and at Step \ref {step:recurse}, there exists $t\in[r]\setminus\{j\}$ such that $R_j$ and $R_t$ are subsets of the vertex-set of a single connected component of $G\setminus S$, while at the subsequent recursive calls, there does not exist such $t$ (by our definition of ${\cal U}_j$). Thus, by Lemma \ref{lem:enumConn}, the running time of each call of \alg{A}$_k$ to \alg{B}$_k$ is bounded by $\OO(n+m)$.\footnote{We remark that by not guessing $i$ in advance, but considering two distinct indices, $j$ and $j'$, in Step \ref{step:recurse}, the algorithm can be modified to run in linear time.} Since $r=\OO(n)$, we have that \alg{A}$_k$ runs in time $\OO(n(n+m))$. The correctness of \alg{A}$_k$ easily follows from Lemmata \ref{lem:enumConn} and \ref{lem:redBlueSmall}.
\end{proof}

Finally, by Theorem \ref{cor:main}, Observation \ref{obs:redBlueCMSO} and Lemma \ref{lem:redBlueUnbreak}, we conclude the correctness of Lemma \ref{lem:redBlueMain}.

\bibliographystyle{siam}
 \bibliography{references1,metakernels_extended,posets-fo}

\section{Appendix} 
Before presenting the proof of  Lemma \ref{lem:break}, we recall the notion of {\em universal sets}.

\begin{definition}\label{dfef:universalSet}
Let $n,k,p\in\mathbb{N}$, and let ${\cal F}$ be a set of functions $f:[n]\rightarrow \{0,1\}$. We say that ${\cal F}$ is an {\em $(n,k,p)$-universal set} if for every subset $I\subseteq [n]$ of size $k$ and a function $f':I\rightarrow\{0,1\}$ that assigns '1' to exactly $p$ indices, there is a function $f\in{\cal F}$ such that for all $i\in I$, $f(i)=f'(i)$.
\end{definition}

The next result asserts that small universal sets can be computed efficiently.

\begin{lemma}[\cite{DBLP:journals/jacm/FominLPS16}]\label{lem:universalSet}
There exists an algorithm that, given $n,k,p\in\mathbb{N}$, computes an $(n,k,p)$-universal set ${\cal F}$ of size $\binom{k}{p}2^{o(k)}\cdot\log n$ in deterministic time $\binom{k}{p}2^{o(k)}\cdot n\log n$.
\end{lemma}

\subsection{Proof of Lemma \ref{lem:break}}\label{app:break}
 
To design the desired algorithm, we first prove two claims.

\begin{claim}\label{claim:break1}
There exists an algorithm that given $s,c\in\mathbb{N}$ and a graph $G$, in time $2^{\OO(c\log(s+c))}\cdot n^3\log n$ either returns an $\displaystyle{(s/2,c)}$-witnessing separation or correctly concludes that there does not exist such a separation, $(X,Y)$, where both $G[X\setminus Y]$ and $G[Y\setminus X]$ contain a connected component of size at least~$s/2$.
\end{claim}

\begin{proof}
For the sake of simplicity, let us identify each vertex in $V(G)$ with a unique integer in $[n]$. Our algorithm works as follows. By using the algorithm in Lemma \ref{lem:universalSet}, it computes (in time $2^{\OO(c\log(s+c))}\cdot n\log n$) an $(n,s+c,c)$-universal set ${\cal F}$ of size $\binom{s+c}{c}2^{o(s+c)}\cdot\log n$. Then, for every $f\in{\cal F}$, it performs the following operations. First, it computes the set $\cal C$ of connected components of $G[f^{-1}(0)]$. Then, for every two distinct connected components $C,C'\in{\cal C}$, it computes a minimum vertex-cut $S$ that is disjoint from $V(C)\cup V(C')$ and which separates $V(C)$ and $V(C')$ (that is, $C$ and $C'$ are subgraphs of different connected components of $G\setminus S$). Notice that this computation can be done by contracting the edges of any spanning tree of $C$ and any spanning tree of $C'$, and then obtaining a minimum vertex-cut between the two resulting vertices. In case $|S|\leq c$, the algorithm returns the following separation $(X,Y)$: the set $X$ contains the union of $S$ and the set of vertices of the connected component of $G\setminus S$ that contains $C$ as a subgraph, and $Y=S\cup (V(G)\setminus X)$. Overall, the total time to perform the operations presented for each individual $f\in{\cal F}$ can be bounded by $\OO(n^3)$ by applying the sparsifying technique of Nagamochi and Ibaraki \cite{NagamochiIbaraki} and the classical Ford-Fulkerson. Finally, if no separation was returned, the algorithm concludes that there does not exist an $\displaystyle{(s/2,c)}$-witnessing separation, $(X,Y)$, where both $G[X\setminus Y]$ and $G[Y\setminus X]$ contain a connected component of size at least~$s/2$.

Clearly, the algorithm runs in time $2^{\OO(c\log(s+c))}\cdot n^3\log n$, and if it returns a separation, then it is an $\displaystyle{(s/2,c)}$-witnessing separation. Next, suppose that there exists an $\displaystyle{(s/2,c)}$-witnessing separation $(X,Y)$ where both $G[X\setminus Y]$ and $G[Y\setminus X]$ contain a connected component of size at least $s/2$. Let $\widehat{C}$ and $\widehat{C}'$ denote a connected component of $G[X\setminus Y]$ of size at least $s/2$ and a connected component of $G[Y\setminus X]$ of size at least $s/2$, respectively. Now, let $\widetilde{C}$ and $\widetilde{C}'$ denote a connected subgraph of $C$ on exactly $\lceil s/2\rceil$ vertices and a connected subgraph of $C$ on exactly $\lceil s/2\rceil$ vertices, respectively. Then, by the definition of an $(n,s+c,c)$-universal set, there exists $f\in{\cal F}$ such that for all $v\in X\cap Y$, $f(v)=1$ and for each $v\in V(C)\cup V(C')$, $f(v)=0$. When the algorithm examines such a function $f$, it holds that $X\cap Y$ is a vertex-cut of that is disjoint from $V(C)\cup V(C')$ and which separates $V(C)$ and $V(C')$, where $C$ and $C'$ are the connected components of $G[f^{-1}(0)]$ that contain $\widetilde{C}$ and $\widetilde{C}'$ as subgraphs, respectively. Then, the algorithm returns an $\displaystyle{(s/2,c)}$-witnessing separation. This concludes the proof of the claim.
\end{proof}

\begin{claim}\label{claim:break2}
There exists an algorithm that given $s,c\in\mathbb{N}$ and a graph $G$, in time $2^{\OO(c\log(s+c))}\cdot n\log n$ either returns an $\displaystyle{(s/2^c,c)}$-witnessing separation or correctly concludes that there does not exist an $\displaystyle{(s,c)}$-witnessing separation $(X,Y)$, where not both $G[X\setminus Y]$ and $G[Y\setminus X]$ contain a connected component of size at least~$s/2$.
\end{claim}

\begin{proof}
For the sake of simplicity, let us identify each vertex in $V(G)$ with a unique integer in $[n]$. Our algorithm works as follows. If $n<2s$, it concludes that there does not exist an $\displaystyle{(s,c)}$-witnessing separation $(X,Y)$. Otherwise, by using the algorithm in Lemma \ref{lem:universalSet}, it computes (in time $2^{\OO(c\log(s+c))}\cdot n\log n$) an $(n,\lfloor3s/2\rfloor+c,c)$-universal set ${\cal F}$ of size $\binom{\lfloor3s/2\rfloor+c}{c}2^{o(s+c)}\cdot\log n$. Then, for every $f\in{\cal F}$, it performs the following operations. First, it computes the set $\cal C$ of connected components of $G[f^{-1}(0)]$. For every $C\in{\cal C}$, denote ${\cal C}_C=\{C'\in{\cal C}: N(V(C'))=N(V(C)), |V(C')| < s/2\}$. Then, if $3s/2<|\bigcup_{C'\in {\cal C}_C}V(C')|$ and as long as this condition holds, by removing one-by-one the largest connected component in $\cal C$, the algorithm ensures that $|\bigcup_{C'\in {\cal C}_C}V(C')|\leq 3s/2$. After handling each $C\in{\cal C}$ individually, if there exists $C\in{\cal C}$ such that $s/2^c\leq|\bigcup_{C'\in {\cal C}_C}V(C')|\leq 3s/2$ and $|N(V(C))|\leq c$, it returns the following separation $(X,Y)$: $X=N[\bigcup_{C'\in {\cal C}_C}V(C')]$, and $Y=N[V(G\setminus\bigcup_{C'\in {\cal C}_C}V(C'))]$. Finally, if no separation was returned, the algorithm concludes that there does not exist an $\displaystyle{(s,c)}$-witnessing separation $(X,Y)$, where both $G[X\setminus Y]$ and $G[Y\setminus X]$ do not contain a connected component of size at least~$s/2$.

Clearly, the algorithm runs in time $2^{\OO(c\log(s+c))}\cdot n\log n$, and if it returns a separation, then it is an  $\displaystyle{(s/2^c,c)}$-witnessing separation. Next, suppose that there exists an $\displaystyle{(s,c)}$-witnessing separation $(X',Y')$ where not both $G[X'\setminus Y']$ and $G[Y'\setminus X']$ contain a connected component of size at least $s/2$. Then, $n\geq 2s$ and there also exists an $\displaystyle{(s,c)}$-witnessing separation $(X,Y)$ where $|X\setminus Y|\leq\lfloor3s/2\rfloor$ and $G[X]$ does not contain a connected component of size at least~$s/2$. Since $|X\cap Y|\leq c$, there exists a subset $S\subseteq X\setminus Y$ of size at least $\displaystyle{|X\setminus Y|}{2^{|X\cap Y|}}\geq s/2^c$ such that $N(S)\subseteq X\cap Y$ and for every two connected components $C$ and $C'$ of $G[S]$, it holds that $N(V(C))=N(V(C'))$. Then, by the definition of an $(n,\lfloor3s/2\rfloor+c,c)$-universal set, there exists $f\in{\cal F}$ such that for all $v\in X\cap Y$, $f(v)=1$ and for each $v\in S$, $f(v)=0$. When the algorithm examines such a function $f$, there exists $C\in{\cal C}$ such that $s/2^c\leq|\bigcup_{C'\in {\cal C}_C}V(C')|\leq 3s/2$ and $|N(V(C))|\leq c$. Then, the algorithm returns an $\displaystyle{(s/2^c,c)}$-witnessing separation. This concludes the proof of the claim.
\end{proof}

To conclude that Lemma \ref{lem:break} is correct, note that for all $x\geq y$, an $\displaystyle{(x,c)}$-witnessing separation is also an $\displaystyle{(y,c)}$-witnessing separation, and that if a graph does not have an $\displaystyle{(s/2^c,c)}$-witnessing separation then it is $(s,c)$-unbreakable. Thus, we apply the algorithms given by Claims~\ref{claim:break1} and \ref{claim:break1}. If at least one of them returns a separation, which is an $\displaystyle{(s/2^c,c)}$-witnessing separation, then we return this separation, and otherwise we correctly conclude that $G$ is $(s,c)$-unbreakable.\qed


\end{document}